\documentclass[11pt]{article}
\usepackage{amsmath,amssymb,amsfonts,latexsym,graphicx,amsthm}
\usepackage{fullpage,color}
\usepackage{url,hyperref}
\usepackage{comment}
\usepackage[linesnumbered,boxed,ruled,vlined]{algorithm2e}
\usepackage{framed}
\usepackage[shortlabels]{enumitem}

\usepackage[margin=1in]{geometry}

\newtheorem{fact}{Fact}
\newtheorem{theorem}{Theorem}[section]
\newtheorem{lemma}{Lemma}[section]
\newtheorem{definition}{Definition}[section]

\newtheorem{assumption}{Assumption}[section]

        {\medskip}

\newcommand{\floor}[1]{\lfloor #1 \rfloor}

\begin{document}

\title{Faster Min-Plus Product for Monotone Instances}
\author{
	Shucheng Chi \thanks{Tsinghua University, \href{}{chisc21@mails.tsinghua.edu.cn}}
	\and
	Ran Duan \thanks{Tsinghua University, \href{}{duanran@mail.tsinghua.edu.cn}} 
	\and
	Tianle Xie \thanks{Tsinghua University, \href{}{xtl21@mails.tsinghua.edu.cn}}
	\and 
	Tianyi Zhang  \thanks{Tel Aviv University, \href{}{tianyiz21@tauex.tau.ac.il}}
}
\date{}

\maketitle

\begin{abstract}
	In this paper, we show that the time complexity of monotone min-plus product of two $n\times n$ matrices is $\tilde{O}(n^{(3+\omega)/2})=\tilde{O}(n^{2.687})$, where $\omega < 2.373$ is the fast matrix multiplication exponent [Alman and Vassilevska Williams 2021]. That is, when $A$ is an arbitrary integer matrix and $B$ is either row-monotone or column-monotone with integer elements bounded by $O(n)$, computing the min-plus product $C$ where $C_{i,j}=\min_k\{A_{i,k}+B_{k,j}\}$ takes $\tilde{O}(n^{(3+\omega)/2})$ time, which greatly improves the previous time bound of $\tilde{O}(n^{(12+\omega)/5})=\tilde{O}(n^{2.875})$ [Gu, Polak, Vassilevska Williams and Xu 2021]. Then by simple reductions, this means the following problems also have $\tilde{O}(n^{(3+\omega)/2})$ time algorithms:
	\begin{itemize}
	    \item $A$ and $B$ are both bounded-difference, that is, the difference between any two adjacent entries is a constant. The previous results give time complexities of $\tilde{O}(n^{2.824})$ [Bringmann, Grandoni, Saha and Vassilevska Williams 2016] and $\tilde{O}(n^{2.779})$ [Chi, Duan and Xie 2022].
	    \item $A$ is arbitrary and the columns or rows of $B$ are bounded-difference. Previous result gives time complexity of $\tilde{O}(n^{2.922})$ [Bringmann, Grandoni, Saha and Vassilevska Williams 2016].
	    \item The problems reducible to these problems, such as language edit distance, RNA-folding, scored parsing problem on BD grammars. [Bringmann, Grandoni, Saha and Vassilevska Williams 2016].
	\end{itemize}
	Finally, we also consider the problem of min-plus convolution between two integral sequences which are monotone and bounded by $O(n)$, and achieve a running time upper bound of $\tilde{O}(n^{1.5})$. Previously, this task requires running time $\tilde{O}(n^{(9+\sqrt{177})/12}) = O(n^{1.859})$ [Chan and Lewenstein 2015].

	%Min-plus product between two $n\times n$ integral matrices is a basic problem in theoretical computer science. The fastest algorithm that computes min-plus products runs in $n^{3-o(1)}$ time, while it is a big open problem whether a truly sub-cubic running time can be achieved. In this paper, we focus on special cases where the input matrices have certain structures.
	%\begin{itemize}
	%	\item First of all, as our main result, we provide a simple and fast algorithm when both input matrices have bounded difference; namely, the difference between any two adjacent entries is at most $1$. Our algorithm runs in time $\tilde{O}(n^{(3+\omega)/2})$, where $\omega < 2.373$ is the fast matrix multiplication exponent. The previous two upper bounds are  $\tilde{O}(n^{2.824})$ by [Bringmann, Grandoni, Saha and Williams, FOCS 2016] and $\tilde{O}(n^{2 + \omega /3})$ by [Chi, Duan and Xie, to appear in SODA 2022].
	
	%	\item Secondly, we generalize our algorithmic framework for bounded-difference inputs to the monotone setting where the first matrix is arbitrary and the second matrix has monotone and bounded rows. Our algorithm runs in time $\tilde{O}(n^{2+\omega/3})$, which improves on the previously best algorithm with running time $\tilde{O}(n^{(12+\omega)/5})$ [Gu, Poloak, Williams and Xu, ICALP 2021]. 
		
	%	\item 
	%\end{itemize}
\end{abstract}

\thispagestyle{empty}
\clearpage
\setcounter{page}{1}
\pagestyle{plain}

\section{Introduction}
The min-plus product $C = A\star B$ between two $n\times n$ matrices $A, B$ is defined as $C_{i, j} = \min_{1\leq k\leq n}\{A_{i,k} + B_{k, j} \}$. The straightforward algorithm for min-plus product runs in $O(n^3)$ time, and a long line of research has been dedicated to breaking this cubic barrier. The currently fastest algorithm by Williams~\cite{williams2018faster} for min-plus product runs in time $n^3 / 2^{\Theta(\sqrt{\log n})}$, and it remains a major open question whether a truly sub-cubic running time of $O(n^{3-\epsilon})$ can be achieved for some constant $\epsilon > 0$. In fact, it is widely believed that truly sub-cubic time algorithms do not exist according to the famous APSP hardness conjecture from the literature of fine-grained complexity~\cite{williams2018some}.

Although min-plus product is hard in general cases, when the input matrices have certain structures, truly sub-cubic time algorithms are known. For example, when all matrix entries are bounded in absolute value by $W$, min-plus product can be computed in time $\tilde{O}(Wn^\omega)$~\cite{alon1997exponent}. Matrices with more general structural properties are studied in recent years. In paper \cite{bringmann2019truly}, the authors introduced the notion of bounded-difference matrices.
\begin{definition}\label{bd-def}
	An integral matrix is called \textbf{bounded-difference}, if each pair of adjacent elements differ by at most a constant $\delta$. Formally, a bounded-difference $n\times n$ matrix $X$ satisfies that for any pair of indices $1\leq i, j\leq n$, we have:
	$$|X_{i, j} - X_{i, j+1}|\leq \delta$$
	$$|X_{i, j} - X_{i+1, j}|\leq \delta$$
\end{definition}
The importance of this special type of min-plus product between bounded-difference matrices is demonstrated by its connection to sub-cubic algorithms for other problems (for example, language edit distance~\cite{bringmann2019truly}, RNA folding~\cite{bringmann2019truly}, and tree edit distance~\cite{mao2021breaking}). 
As their main technical result, the authors of~\cite{bringmann2019truly} gave the first sub-cubic time algorithm for computing min-plus product between two $n\times n$ bounded-difference matrices in time $\tilde{O}(n^{2.824})$. This upper bound was improved significantly to $\tilde{O}(n^{2 + \omega/3})$ by a very recent work~\cite{chi2022faster}; here $\omega$ refers to the fast matrix multiplication exponent~\cite{alman2021refined}.

Following \cite{bringmann2019truly}, less restricted types of matrices are studied in~\cite{williams2020truly,gu2021faster}. In their work~\cite{williams2020truly}, Williams and Xu considered the case where one of the input matrices is monotone.
\begin{definition}\label{monotone-def}
	An $n\times n$ integral matrix is called \textbf{row-monotone}, or simply \textbf{monotone}, if all entries are nonnegative integers bounded by $O(n)$ and each row of this matrix is non-decreasing, that is, if $X$ is monotone, then for $i,j$, $0\leq X_{i, j} = O(n), X_{i, j}\leq X_{i, j+1}$. Similarly we can define \textbf{column-monotone} matrix. 
\end{definition}
It was shown in~\cite{gu2021faster} that min-plus product in the bounded-difference setting can be reduced to the monotone setting in quadratic time, so this monotone setting is at least as hard in general. With this definition, Williams and Xu~\cite{williams2020truly} studied the monotone min-plus product problem where $A$ is an arbitrary integral matrix and $B$ is monotone, which has an application in the batch range mode problem, and they presented a sub-cubic algorithm with running time $\tilde{O}(n^{(15+\omega)/6})$. This upper bound was later improved to $\tilde{O}(n^{(12+\omega)/5})$ in a recent work \cite{gu2021faster}.

Other than matrix pairs, the concept of min-plus also applies to sequence pairs. Given two sequences $A, B$ with $n$ entries, their min-plus convolution $C = A\diamond B$ can be defined as $C_{k} = \min_{i=1}^{k-1}\{A_i + B_{k-i}\}, \forall 2\leq k\leq 2n$. Chan and Lewenstein~\cite{chan2015clustered} studied fast algorithms for min-plus convolution when input sequences $A, B$ are monotone.
\begin{definition}\label{seq-monotone-def}
	An integral sequence of length $n$ is called \textbf{monotone}, if this sequence is monotonically increasing, plus that all entries are nonnegative and bounded by $O(n)$.
\end{definition}
When both sequences $A, B$ are monotone, Chan and Lewenstein~\cite{chan2015clustered} showed that min-plus convolution can be computed in sub-quadratic time $\tilde{O}(n^{(9+\sqrt{177})/12}) = O(n^{1.859})$. This problem is important due to its connections with other problems like histogram indexing and necklace alignment~\cite{amir2014hardness,chan2015clustered,bremner2006necklaces}. 
%CITATION

\subsection{Our results}
The main result of this paper is a faster algorithm for min-plus matrix product in the monotone setting.
\begin{theorem}\label{matrix}
	There is a randomized algorithm that computes min-plus product $A\star B$  with expected running time $\tilde{O}(n^{(3+\omega)/2})$, where $A$ is an $n\times n$ integral matrix while $B$ is an $n\times n$ monotone matrix.
\end{theorem}
This improves on the previous upper bound of $\tilde{O}(n^{(12+\omega)/5})$~\cite{gu2021faster}; as a corollary, by a reduction from the bounded-difference setting to the monotone setting, this also implies that min-plus matrix product between two bounded-difference matrices can be computed in time $\tilde{O}(n^{(3+\omega)/2})$, which improves upon the recent upper bound of $\tilde{O}(n^{2+\omega/3})$~\cite{chi2022faster}.

By adapting our techniques to the monotone min-plus convolution problem, we can achieve the following result: 
%for comparison, the state-of-the-art algorithm for monotone min-plus convolution has sub-quadratic time $\tilde{O}(n^{(9+\sqrt{177})/12}) = O(n^{1.859})$~\cite{chan2015clustered}.
\begin{theorem}\label{conv}
	There is a randomized algorithm that computes min-plus convolution between two monotonically increasing integral sequences $A, B$, where entries of $A, B$ are nonnegative integers bounded by $O(n)$, and the expected running time of this algorithm is $\tilde{O}(n^{1.5})$.
\end{theorem}

In the appendix, we also generalize Theorem~\ref{matrix} to column-monotone $B$:
\begin{theorem}\label{app-matrix}
	There is a randomized algorithm that computes min-plus product $A\star B$  with expected running time $\tilde{O}(n^{(3+\omega)/2})$, where $A$ is an $n\times n$ integral matrix while $B$ is an $n\times n$ column-monotone matrix.
\end{theorem}
Since $(A\star B)^T=B^T\star A^T$, these also solve the case that $A$ is row-monotone or column-monotone and $B$ is arbitrary.

\subsection{Technical overview}
In this subsection, we take an overview of our algorithm for monotone matrix min-plus product. The basic algorithmic framework follows the main idea of the previous work \cite{chi2022faster} but with some important modification so that it can achieve a running time of $\tilde{O}(n^{2+\omega/3})$ for monotone min-plus product instead of bounded-difference min-plus product. To push it down to $\tilde{O}(n^{(3+\omega)/2})$ as stated in Theorem~\ref{matrix}, we need to follow a certain recursive paradigm. For simplicity, let us assume for now that $\omega=2$.\\

\noindent\textbf{The basic algorithm.} Similar to \cite{gu2021faster}, as the first step we take the approximation matrices $\tilde{A}, \tilde{B}$ of the input $A, B$, which are defined as $\tilde{A}_{i, j} = \floor{A_{i, j} / n^{1/3}}$ and $\tilde{B}_{i, j} = \floor{B_{i, j} / n^{1/3}}$, respectively, and then compute $\tilde{C} = \tilde{A}\star\tilde{B}$ using an elementary combinatorial method which takes time $\tilde{O}(n^{8/3})$. (See Section~\ref{mono-basic}.)

The approximation matrix $\tilde{C}$ gives a necessary condition for witness indices $k$ such that $A_{i, k} + B_{k, j} = C_{i, j}$: if the equality holds, then it must be the case that $\tilde{A}_{i, k} + \tilde{B}_{k, j} - \tilde{C}_{i, j} = O(1)$. Using this fact, build the following two polynomial matrices $A(x, y), B(x, y) $ on variables $x, y$:
$$A_{i, k}(x, y) = x^{A_{i, k} - n^{1/3}\cdot \tilde{A}_{i, k}}\cdot y^{\tilde{A}_{i, k}}$$
$$B_{k, j}(x, y) = x^{B_{k, j} - n^{1/3}\cdot \tilde{B}_{k, j}}\cdot y^{\tilde{B}_{k, j}}$$
Suppose we can directly compute $C(x, y) = A(x, y)\cdot B(x, y)$ under the standard notion of $(+, \times)$ of matrix product. Then, to search for the true value $C_{i, j} = \min_k\{A_{i, k} + B_{k, j}\}$, we only need to look at terms $x^c y^d$ of polynomial $C_{i, j}(x, y)$ such that $|d - \tilde{C}_{i, j}| = O(1)$, and determine $C_{i, j}$ to be the minimum over all values of $c + n^{1/3}d$.

Unfortunately, computing $C(x, y) = A(x, y)\cdot B(x, y)$ is very costly in general since the degrees of $y$ can be very large. To reduce the $y$-degrees, the idea is to take $p$-modulo on the exponent of $y$, where $p = \Theta(n^{1/3})$ is a random prime number. Formally, construct two polynomial matrices $A^p(x, y), B^p(x, y)$ as following:
$$A_{i, k}^p(x, y) = x^{A_{i, k} - n^{1/3}\cdot \tilde{A}_{i, k}}\cdot y^{\tilde{A}_{i, k}\mod p}$$
$$B_{k, j}^p(x, y) = x^{B_{k, j} - n^{1/3}\cdot \tilde{B}_{k, j}}\cdot y^{\tilde{B}_{k, j}\mod p}$$
In this way, matrix product $C^p = A^p\cdot B^p$ only requires running time $\tilde{O}(n^{8/3})$. The problem with this approach is that, when we go over all the terms $x^cy^d$ of polynomial $C_{i, j}(x, y)$ such that $|d - \tilde{C}_{i, j}\mod p| = O(1)$, $c+n^{1/3}d$ might be an underestimate of $C_{i, j}$; in fact, it could be the case that for some index $k$, we have:
$$\begin{aligned}
&c = A_{i, k} - n^{1/3}\cdot \tilde{A}_{i, k} + B_{k, j} - n^{1/3}\cdot \tilde{B}_{k, j}\\
&d \equiv \tilde{A}_{i, k} + \tilde{B}_{k, j} \mod p\\
&d \neq \tilde{A}_{i, k} + \tilde{B}_{k, j}
\end{aligned}$$
To resolve this issue, we should first enumerate all triples $i, j, k$ such that $d \equiv \tilde{A}_{i, k} + \tilde{B}_{k, j} \mod p$ and $d \neq \tilde{A}_{i, k} + \tilde{B}_{k, j}$, and then subtract the erroneous terms $x^cy^d$ from $C_{i, j}(x, y)$. To upper bound the total running time, the key point is that when $p$ is a random prime, the probability that $d \equiv \tilde{A}_{i, k} + \tilde{B}_{k, j} \mod p$ is at most $\tilde{O}(1/p)$ when $d\neq \tilde{A}_{i, k} + \tilde{B}_{k, j}$, and therefore the expected number of erroneous terms is bounded by $\tilde{O}(n^3/p) = \tilde{O}(n^{8/3})$.\\

\noindent\textbf{Improvement by recursion.} To push the upper bound exponent from $8/3$ to $2.5$, we again follow the idea in \cite{chi2022faster} of using recursions. Roughly speaking, we will apply a numerical scaling technique on the input matrices $A, B$, and the key technical point is that throughout different numerical scales we need to carefully maintain all erroneous terms.

More specifically, take a random prime $p$ in $[n^{0.5}, 2n^{0.5}]$, and define $A^{(l)}_{i, j} = \floor{(A_{i, j}\mod p)/2^l}$, $B^{(l)}_{i, j} = \floor{(B_{i, j} \mod p) / 2^l}$, $C^{(l)} = \floor{(C_{i, j}\mod p) / 2^l}$, then we will iteratively compute all $C^{(l)}$ with $l = h, h-1, h-2, \cdots, 0$, for some parameter $h$; note that in general $C^{(l)}\neq A^{(l)}\star B^{(l)}$, so computing $C^{(l)}$ would also require information from the original input matrices $A, B$. Once we have $C^{(0)} = C\mod p$, we can deduce the true value of $C$ from the approximation matrix $C^* = A^*\star B^*$, where $A^*_{i, j} = \floor{A_{i, j} / p}$ and $B^*_{i, j} = \floor{B_{i, j} / p}$; note that computing $C^*$ takes time $\tilde{O}(n^{2.5})$.

To compute $C^{(l)}$, the algorithm uses $C^{(l+1)}$ as the approximation matrix. Namely, similar to the basic algorithm, let us construct two $n\times n$ polynomial matrices $A^p, B^p$ on variables $x, y$ in the following way:
$$A^p_{i, k} = x^{A^{(l)}_{i, k} - 2A^{(l+1)}_{i, k}}\cdot y^{ A^{(l+1)}_{i, k}}$$
$$B^p_{k, j} = x^{B^{(l)}_{k, j} - 2B^{(l+1)}_{k, j}}\cdot y^{B^{(l+1)}_{k, j}}$$
Then, compute the standard $(+, \times)$ matrix multiplication $C^p = A^p\cdot B^p$ using fast matrix multiplication. The advantage of numerical scaling is that the degree of $x$ is $0$ or $1$, so polynomial matrix multiplication only takes time $\tilde{O}(n^{2.5})$.

To retrieve $C^{(l)}_{i, j}$, we will prove that $C^{(l)}_{i, j}$ must be equal to some $A^{(l)}_{i, k} + B^{(l)}_{k, j}$ such that the following two conditions hold:
\begin{itemize}
\item $|A^{(l+1)}_{i, k} + B^{(l+1)}_{k, j} - C^{(l+1)}_{i, j}| = O(1)$.
\item $A^*_{i, k} + B^*_{k, j} = C^*_{i, j}$.
\end{itemize}

So, we only need to look at the monomials in $C^p_{i, j}$ whose $y$-degree differs from $C^{(l+1)}_{i, j}$ by at most $O(1)$. However, before this we need to subtract all erroneous terms from $C^{p}_{i, j}$, which are all of those triples $(i, j, k)\in [n]^3$ such that:
\begin{itemize}
	\item $|A^{(l+1)}_{i, k} + B^{(l+1)}_{k, j} - C^{(l+1)}_{i, j}| = O(1)$.
	\item $A^*_{i, k} + B^*_{k, j} \neq C^*_{i, j}$.
\end{itemize}
To efficiently enumerate these triples, the key idea is to maintain them iteratively for all $l = h, h-1, \cdots, 0$ as well, along with the approximation matrices $C^{(l)}$. A technical issue is that the total number of such triples might be as large as $\Theta(n^3)$. This is where we utilize the monotone property (of $B$ and $C$) by grouping consecutive triples into segments such that the total number of segments is bounded by $O(n^3 / p) = O(n^{2.5})$.

\section{Preliminaries}\label{sec:pre}
\noindent\textbf{Notations.} For any integers $a, m$, let $(a\mod m)$ refer to the unique value $b\in \{0, 1, 2, \cdots, m-1\}$ such that $a\equiv b\mod m$. For any positive integer $x$, $[x]$ refers to the set $\{1, 2, 3, \cdots, x\}$. For a matrix $A$ and a real number $x$, $A+x$ means adding $x$ to every element of $A$.\\

\noindent\textbf{Segment trees.} Let $X = \{x_1, x_2, \cdots, x_N \}$ be an integral sequence of $N$ elements which undergoes updates and queries. Each update operation specifies an interval $[i, j]$ and an integer value $u$, then for each $i\leq l\leq j$, $x_l$ is updated as $x_l\leftarrow \min\{x_l, u \}$. Each query operation inspects the current value of an arbitrary element $x_i$. Using standard segment tree data structures~\cite{de2000more},
%CITATION
both update and query operations are supported in $O(\log N)$ deterministic worst-case time.\\

\noindent\textbf{Matrix multiplication.} We denote with $O(n^\omega)$ the arithmetic complexity of multiplying two $n\times n$ matrices. Currently the best bound is $\omega < 2.37286$~\cite{alman2021refined,le2014powers,williams2012multiplying}.\\

\noindent\textbf{Polynomial matrices.} Our algorithm will work with multivariate polynomials. For bivariate polynomials on variables $x, y$, suppose the maximum degrees of $x, y$ are bounded in absolute value by $d_1, d_2$, respectively (we allow their degrees to be negative). Given two polynomials $p, q\in \mathbb{Z}[x, y]$, we can add and subtract $p, q$ in $O(d_1d_2)$ time, and multiply $p, q$ in $\tilde{O}(d_1d_2)$ time using fast-Fourier transformations~\cite{Schnhage2005}. Similar bounds hold for polynomials on three variables $x, y, z$ as well.

We will also work with polynomial matrices from $(\mathbb{Z}[x, y])^{n\times n}$. Products between two matrices in $(\mathbb{Z}[x, y])^{n\times n}$ can be performed as usual, but since each arithmetic operation takes time $\tilde{O}(d_1d_2)$, the cost of matrix multiplication takes time $\tilde{O}(d_1d_2n^\omega)$. To do this, we can reduce it to multiplication of polynomial univariate matrices: replace $y = x^{10d_1}$ and multiply the two univariate matrices, and then take $10d_1$-modulo on the degrees to recover the original degrees of $x, y$ of each element.\\

\noindent\textbf{Distribution of primes.} Let $\pi(x)$ be the prime-counting function that gives the number of primes less than or equal to $x$. According to the famous prime number theorem~\cite{jameson2003prime},
%CITATION
$\pi(x)\sim x / \ln(x)$. As a corollary, for any large enough integer $N$, the number of primes in the range $[N, 2N]$ is at least $\Omega(N / \log N)$.\\

\noindent\textbf{Assumptions and Reductions.} When computing the min-plus product of $A$ and $B$, it is easy to see the following operations will not affect the complexity of computation:
\begin{enumerate}[(a)]
    \item\label{adjust1} We can add the same value to all elements in a row of $A$ or add the same value to all elements in a column of $B$. To recover the original result $A\star B$ from the new result $C$, simply subtract the same value in the corresponding row of $C$ or subtract the same value in the corresponding column of $C$, resp.
    \item\label{adjust2} We can add the same value $\delta$ to all elements in $i-$th column of $A$ \textbf{and} subtract $\delta$ from all elements in $i-$th row of $B$. The min-plus product remain unchanged.
    \item If $B$ is column-monotone, we can make $A$ row-monotone (reverse order), since when $B_{k,j}\leq B_{k+1,j}$, if $A_{i,k}< A_{i,k+1}$, then $A_{i,k+1}+B_{k+1,j}$ cannot be a candidate of $C_{i,j}$, so we can make $A_{i,k+1}\leftarrow A_{i,k}$.
    \item\label{bounded-monotone}\cite{gu2021faster} If $B$ is $\delta$-row-bounded-difference, that is, $|B_{i,j}-B_{i,j+1}|\leq \delta$, then we can add $j\cdot\delta$ to the $j-$th column of $B$ to make $B$ row-monotone, so row-bounded-difference can be reduced to row-monotone. Similarly, if $B$ is $\delta$-column-bounded-difference, it can be reduced to column-monotone (with the change of $A$ by~\ref{adjust2}).
    \item\label{A-bound} If all elements in $B$ are between 0 and $c\cdot n$ for some constant $c$, by~\ref{adjust1}, we can adjust rows of $A$ so that the first column of $A$ are all set to $c\cdot n$, then all elements of $A$ can be made in the range $[0,2c\cdot n]$.
\end{enumerate}

Also, it is easy to get the following fact:
\begin{fact}\label{fact:monotone}
In $C=A\star B$, if $B$ is row-monotone, then $C$ is also row-monotone.
\end{fact}

From~\ref{bounded-monotone} we can reduce $A\star B$ for any $A,B$ to the case that $B$ is row-monotone or column-monotone without $O(n)$-bound, so the general case of $B$ monotone is APSP-hard~\cite{Williams2010}. Thus, from~\ref{A-bound}, we only consider the case that $B$ is row-monotone or column-monotone and all elements in $A$ and $B$ are nonnegative integers bounded by $O(n)$. In this paper, monotone matrices are defined to have this element bound of $O(n)$ as in Definition~\ref{monotone-def}.

\section{Monotone min-plus product}\label{monotone-main}

\subsection{Basic Algorithm}\label{mono-basic}
In this section we prove Theorem~\ref{matrix}, that is, $B$ is row-monotone. Take a constant parameter $\alpha \in (0, 1)$ which is to be determined in the end; for convenience let us assume $n^\alpha$ is an integer.
The algorithm consists of three phases.\\

\noindent\textbf{Approximation.} Define two $n\times n$ integer matrices $\tilde{A}, \tilde{B}$ such that $\tilde{A}_{i, j} = \lfloor A_{i, j} / n^\alpha\rfloor$, $\tilde{B}_{i, j} = \lfloor B_{i, j} / n^\alpha\rfloor$. Therefore, $\tilde{B}$ is an integer matrix whose entries are bounded by $O(n^{1-\alpha})$, and each row of $\tilde{B}$ is non-decreasing. 

Next, compute the approximation matrix $\tilde{C} = \tilde{A}\star\tilde{B}$ in the following way. Initialize each entry of $\tilde{C}$ to be $\infty$, and maintain each row of $\tilde{C}$ using a segment tree that supports interval updates. Then, for every pair of indices $i, k\in [n]$, run the following iterative procedure that scans the $k$-th row of $\tilde{B}$. Starting with index $j = 1$, find the largest index $j\leq j_1\leq n$ such that $\tilde{B}_{k, j} = \tilde{B}_{k, j+1} = \cdots = \tilde{B}_{k, j_1}$ using binary search. Then, update all elements $\tilde{C}_{i, l}\leftarrow \min\{\tilde{C}_{i, l}, \tilde{A}_{i, k} + \tilde{B}_{k, l}\}$ for all $j\leq l\leq j_1$ using the segment tree data structure; notice that this operation is legal since all $\tilde{B}_{k, l}$ are equal when $j\leq l\leq j_1$. After that, set $j \leftarrow j_1+1$ and repeat until $j>n$.\\

\noindent\textbf{Polynomial matrix multiplication.} Uniformly sample a random prime number $p$ in the range $[n^\alpha, 2n^\alpha]$. Construct two polynomial matrices $A^p$ and $B^p$ on variables $x, y$ in the following way:
$$A^p_{i, k} = x^{A_{i, k} - n^\alpha\tilde{A}_{i, k}}\cdot y^{\tilde{A}_{i, k}\mod p}$$
$$B^p_{k, j} = x^{B_{k, j} - n^\alpha\tilde{B}_{k, j}}\cdot y^{\tilde{B}_{k, j}\mod p}$$
Then, compute the standard $(+, \times)$ matrix multiplication $C^p = A^p\cdot B^p$ using fast matrix multiplication algorithms.\\

\noindent\textbf{Subtracting erroneous terms.} The last phase is to extract the true values $C_{i, j}$'s from $\tilde{C}$ and $C^p$. The algorithm iterates over all offsets $b\in \{0, 1, 2\}$, and computes the set $T_b\subseteq [n]^3$ of all triples of indices $(i, j, k)$ such that $\tilde{A}_{i, k} + \tilde{B}_{k, j} \neq \tilde{C}_{i, j}+b$ but $\tilde{A}_{i, k} + \tilde{B}_{k, j} \equiv \tilde{C}_{i, j}+b\mod p$; in the running time analysis, we will show that $T_b$ can be computed in time $\tilde{O}(|T_b| + n^{3-\alpha})$.

For each pair of indices $i, j\in [n]$, collect all the non-zero monomials $\lambda x^cy^d$ (for some integer $\lambda$) of $C_{i, j}^p$ such that 
$$d\equiv \tilde{C}_{i, j} + b\mod p$$ 
and let $C_{i, j, b}^p(x)$ be the sum of all such terms $\lambda x^c$. Next, compute a polynomial 
$$R_{i, j, b}^p(x) = \sum_{(i, j, k)\in T_b} x^{A_{i, k} - n^\alpha\tilde{A}_{i, k} + B_{k, j} - n^\alpha\tilde{B}_{k, j}}$$
Finally, let $s_{i, j, b}$ be the minimum degree of $x$ of the polynomial $C^p_{i, j, b}(x) - R^p_{i, j, b}(x)$, and compute a candidate value $c_{i, j, b} = n^\alpha (\tilde{C}_{i, j} + b) + s_{i, j, b}$. Ranging over all integer offsets $b\in \{0, 1, 2\}$, take the minimum of all candidate values and output as $C_{i, j} = \min_{0\leq b\leq 2}\{c_{i, j, b}\}$.

\subsubsection{Proof of correctness}
\begin{lemma}
	For any triple $(i, j, k)\in [n]^3$ such that $A_{i, k} + B_{k, j} = C_{i, j}$, we have $$0\leq \tilde{A}_{i, k} + \tilde{B}_{k, j} - \tilde{C}_{i, j}\leq 2$$
\end{lemma}
\begin{proof}
	Clearly $\tilde{A}_{i, k} + \tilde{B}_{k, j} - \tilde{C}_{i, j}\geq 0$, so we only need to focus on the second inequality.
	
	Suppose $\tilde{C}_{i, j} = \tilde{A}_{i, l} + \tilde{B}_{l, j}$ for some $l$. Then, by definition of $\tilde{A}, \tilde{B}$, we have:
	$$\begin{aligned}
		n^\alpha\tilde{C}_{i, j} &= n^\alpha\tilde{A}_{i, l} + n^\alpha\tilde{B}_{l, j} \geq A_{i, l} + B_{l, j} - 2n^\alpha\geq C_{i, j} - 2n^\alpha\\
		&= A_{i, k} + B_{k, j} - 2n^\alpha\geq n^\alpha\tilde{A}_{i, k} + n^\alpha\tilde{B}_{k, j} - 2n^\alpha
	\end{aligned}$$
	Hence, $\tilde{A}_{i, k} + \tilde{B}_{k, j} - \tilde{C}_{i, j}\leq 2$.
\end{proof}

Next we argue that our algorithm correctly computes all entries $C_{i, j}$. Let $l$ be the index such that $C_{i, j} = A_{i, l} + B_{l, j}$. By the above lemma, there exists an integer offset $b\in \{0, 1, 2\}$ such that $\tilde{A}_{i, l} + \tilde{B}_{l, j} = \tilde{C}_{i, j} + b$. Therefore, by construction of polynomial matrices $A^p, B^p$, we have:
$$\begin{aligned}
	C^p_{i, j, b}(x) &= \sum_{k\mid \tilde{A}_{i, k} + \tilde{B}_{k, j} = \tilde{C}_{i, j} + b}x^{A_{i, k} - n^\alpha \tilde{A}_{i, k} + B_{k, j} - n^\alpha \tilde{B}_{k, j}}\\
	&+\sum_{k\mid (\tilde{A}_{i, k} + \tilde{B}_{k, j} \neq \tilde{C}_{i, j} + b )\wedge (\tilde{A}_{i, k} + \tilde{B}_{k, j} \equiv \tilde{C}_{i, j} + b\mod p)}x^{A_{i, k} - n^\alpha \tilde{A}_{i, k} + B_{k, j} - n^\alpha \tilde{B}_{k, j}}\\
	&= \sum_{k\mid \tilde{A}_{i, k} + \tilde{B}_{k, j} = \tilde{C}_{i, j} + b}x^{A_{i, k} - n^\alpha \tilde{A}_{i, k} + B_{k, j} - n^\alpha \tilde{B}_{k, j}} +\sum_{(i, j, k)\in T_b}x^{A_{i, k} - n^\alpha \tilde{A}_{i, k} + B_{k, j} - n^\alpha \tilde{B}_{k, j}}\\
	&= x^{-n^\alpha (\tilde{C}_{i, j} + b)} \cdot\sum_{k\mid \tilde{A}_{i, k} + \tilde{B}_{k, j} = \tilde{C}_{i, j} + b}x^{A_{i, k}+ B_{k, j}} + R^p_{i, j, b}(x)
\end{aligned}$$
Therefore, $$x^{-n^\alpha (\tilde{C}_{i, j} + b)} \cdot\sum_{k\mid \tilde{A}_{i, k} + \tilde{B}_{k, j} = \tilde{C}_{i, j} + b}x^{A_{i, k}+ B_{k, j}}= C_{i, j, b}^p(x) - R_{i, j, b}^p(x)$$ 
Since $\tilde{A}_{i, l} + \tilde{B}_{l, j} = \tilde{C}_{i, j} + b$, we can extract $C_{i, j}$ from terms of $C_{i, j, b}^p(x) - R_{i, j, b}^p(x)$. In the other way, every nonzero term $C_{i, j, b}^p(x) - R_{i, j, b}^p(x)$ corresponds to a sum of $A_{i,k}+B_{k,j}$, which is at least $C_{i,j}$.

\subsubsection{Running time analysis}
\begin{lemma}
	Computing the approximation matrix $\tilde{C}$ takes time $\tilde{O}(n^{3-\alpha})$.
\end{lemma}
\begin{proof}
	For any pair of $i, k$, the algorithm iteratively increases index $j$ and apply update operations on the segment tree data structure. Since elements of $B$ are bounded by $O(n)$, the total number of different values on the $k$-th row of $\tilde{B}$ is at most $O(n^{1-\alpha})$. Therefore, the number of iterations over $j$ is at most $O(n^{1-\alpha})$ as well. Hence, the running time of this phase is $\tilde{O}(n^{3-\alpha})$.
\end{proof}

As for polynomial matrix multiplication, by definition the $x$-degree and $y$-degree of $A^p, B^p$ are both bounded by $O(n^\alpha)$ in absolute value, so the matrix multiplication takes time $O(n^{\omega + 2\alpha})$.
\begin{lemma}
	The triple set $T_b$ can be computed in time $\tilde{O}(|T_b| + n^{3-\alpha})$.
\end{lemma}
\begin{proof}
	Fix any pair of $i, k$, we try to find all $j$ such that $(i, j, k)\in T_b$. By Fact~\ref{fact:monotone}, $\tilde{B}$ and $\tilde{C}$ are both row-monotone, so we can divide the $k$-th row of $\tilde{B}$ and $i$-th row of $\tilde{C}$ into at most $O(n^{1-\alpha})$ consecutive intervals, such that entries in each interval are all equal.
	So there are $O(n^{1-\alpha})$ intervals $[j_0,j_1]$ such that for all $j\in[j_0,j_1]$, $\tilde{A}_{i, k} + \tilde{B}_{k, j}$ and $\tilde{C}_{i,j}$ are fixed. Therefore, as the total number of such row intervals is bounded by $O(n^{3-\alpha})$, the total running time becomes $\tilde{O}(|T_b| + n^{3-\alpha})$.
	%For each such row interval $[j_0, j_1]$, we need to identify all $j\in [j_0, j_1]$ such that $\tilde{A}_{i, k} + \tilde{B}_{k, j} -b\neq \tilde{C}_{i, j}$ and $\tilde{A}_{i, k} + \tilde{B}_{k, j}-b\equiv \tilde{C}_{i, j}\mod p$. The advantage here is that $\tilde{A}_{i, k}+\tilde{B}_{k, j}-b$ is a fixed value when $j\in [j_0, j_1]$, so this task can be done in time proportional to the output size (with extra log factors) by storing a binary search tree on all values $\{\tilde{C}_{i, j}\mid j\in [j_0, j_1] \}$.
\end{proof}

By the above lemma, the subtraction phase takes time $\tilde{O}(|T_b| + n^{3-\alpha})$ as well. So it suffices to bound the size of $T_b$. For any $(i, j, k)\in [n]^3$ such that $\tilde{A}_{i, k} + \tilde{B}_{j, k}\neq \tilde{C}_{i, j}+b$ 
since $|\tilde{A}_{i, k} + \tilde{B}_{j, k} - \tilde{C}_{i, j}- b|$ is bounded by $O(n)$, there are at most $O(1/\alpha)=O(1)$ different primes in $[n^\alpha, 2n^\alpha]$ that divides $\tilde{A}_{i, k} + \tilde{B}_{j, k} - \tilde{C}_{i, j}- b$.
Since $p$ is a uniformly random prime in the range $[n^\alpha, 2n^\alpha]$, the probability that $\tilde{A}_{i, k} + \tilde{B}_{j, k} - \tilde{C}_{i, j}- b$ can be divided by $p$ is bounded by $\tilde{O}(n^{-\alpha})$. Hence, by linearity of expectation, we have $\mathbb{E}_p[|T_b|]\leq \tilde{O}(n^{3-\alpha})$.

Throughout all three phases, the expected running time of our algorithm is bounded by $\tilde{O}(n^{3-\alpha} + n^{\omega+2\alpha})$. Taking $\alpha = 1 - \omega / 3$, the running time becomes $\tilde{O}(n^{2+\omega /3})$.

\subsection{Recursive Algorithm}\label{sec:recursive}

Let $\alpha\in (0, 1)$ be a constant parameter to be determined later, and pick a uniformly random prime number $p$ in the range of $[40n^\alpha, 80n^\alpha]$. Without loss of generality, let us assume that $n$ is a power of $2$. Next we make the following assumption about elements in $A$ and $B$:

\begin{assumption}\label{basic-assumption}
For every $i,j$, either $(A_{i,j}\mod p)<p/3$ or $A_{i,j}=+\infty$. For every $B_{i,j}$, $(B_{i,j}\mod p)<p/3$. And each row of $B$ is monotone.
\end{assumption}

\begin{lemma}\label{assumption-correctness}
The general computation of $A\star B$ where $B$ is row-monotone can be reduced to a constant number of computations of $A^i\star B^i$, where all of $A^i,B^i$'s satisfy Assumption~\ref{basic-assumption}.
\end{lemma}
\begin{proof}
The idea is very simple: for every element $A_{i,j}$, 
\begin{itemize}
    \item if $(A_{i,j}\mod p)<p/3$, $A'_{i,j}=A_{i,j}$, $A''_{i,j}=A'''_{i,j}=+\infty$
    \item if $p/3<(A_{i,j}\mod p)<2p/3$, $A''_{i,j}=A_{i,j}$, $A'_{i,j}=A'''_{i,j}=+\infty$
    \item if $(A_{i,j}\mod p)>2p/3$, $A'''_{i,j}=A_{i,j}$, $A'_{i,j}=A''_{i,j}=+\infty$
\end{itemize}

When we try to define $B',B''$ and $B'''$ similarly, to make them still row-monotone, we need to fill the ``blanks'' with appropriate numbers.
\begin{itemize}
    \item if $(B_{i,j}\mod p)<p/3$, let $B'_{i,j}=B_{i,j}$ and $B''_{i,j}=p\cdot\lfloor B_{i,j}/p\rfloor+\lceil p/3\rceil$, $B'''_{i,j}=p\cdot\lfloor B_{i,j}/p\rfloor+\lceil 2p/3\rceil$
    \item if $p/3<(B_{i,j}\mod p)<2p/3$, let $B''_{i,j}=B_{i,j}$ and $B'_{i,j}=p\cdot\lfloor B_{i,j}/p+1\rfloor$, $B'''_{i,j}=p\cdot\lfloor B_{i,j}/p\rfloor+\lceil 2p/3\rceil$
    \item if $(B_{i,j}\mod p)>2p/3$, let $B'''_{i,j}=B_{i,j}$ and $B'_{i,j}=p\cdot\lfloor B_{i,j}/p+1\rfloor$, $B''_{i,j}=p\cdot\lfloor B_{i,j}/p+1\rfloor+\lceil p/3\rceil$
\end{itemize}
We can see each pair of $A^*$ and $B^*$, where $A^*\in\{A',A''-\lceil p/3\rceil, A'''-\lceil 2p/3\rceil\}$, $B^*\in\{B',B''-\lceil p/3\rceil, B'''-\lceil 2p/3\rceil\}$,
all satisfy Assumption~\ref{basic-assumption}, so we compute $C'=\min_{A^*\in\{A',A'',A'''\}\atop B^*\in\{B',B'',B'''\}}\{A^*\star B^*\}$ (element-wise minimum). Since elements in $B',B'',B'''$ become no smaller than the corresponding ones in $B$, similarly for $A', A'',A'''$, so $C'_{i,j}\geq C_{i,j}$. But for the $k$ satisfying $A_{i,k}+B_{k,j}=C_{i,j}$, $A_{i,j}$ and $B_{k,j}$ must be in one of the 9 pairs, so $C'_{i,j}=C_{i,j}$.
\end{proof}

Define integer $h$ such that $2^{h-1}\leq p< 2^h$. For each integer $0\leq l\leq h$, let $A^{(l)}$ be the $n\times n$ matrix defined as $A^{(l)}_{i, j} = \lfloor \frac{A_{i,j}\mod p}{ 2^{l}} \rfloor$ if $A_{i,j}$ is finite, otherwise $A^{(l)}_{i,j}=+\infty$, similarly define matrix $B^{(l)}= \lfloor \frac{B_{i,j}\mod p}{ 2^{l}} \rfloor$. 

 Define $A^*$ and $B^*$ as $A^*_{i,j}=\lfloor A_{i,j}/p\rfloor$ and $B^*_{i,j}=\lfloor B_{i,j}/p\rfloor$. We use the segment tree structure to calculate $C^*=A^*\star B^*$ in $\tilde{O}(n^{3-\alpha})$ time. By Assumption~\ref{basic-assumption}, $C^*_{i,j}=\lfloor C_{i,j}/p\rfloor$ if $C_{i,j}$ is finite.

We will recursively calculate $C^{(l)}$ for $l=h,h-1,\cdots,0$. Intuitively, $C^{(l)}_{i,j}$ is the approximate result obtained from $A^{(l)}_{i,k}$ and $B^{(l)}_{k,j}$ for those $k$ satisfying $C^*_{i,j}=A^*_{i,k}+B^*_{k,j}$. If $C_{i,j}$ is finite, $C^{(l)}$ will satisfy that
\begin{enumerate}[(1)]
    \item $\lfloor \frac{(C_{i,j} \mod p)-2(2^l-1)}{ 2^{l}} \rfloor \leq C^{(l)}_{i, j} \leq \lfloor \frac{(C_{i,j}\mod p)+2(2^l-1)}{ 2^{l}} \rfloor$
    \item If $C^*_{i,j_0}=C^*_{i,j_1}$ for $j_0<j_1$, the elements in $C^{(l)}_{i,j_0},\cdots, C^{(l)}_{i,j_1}$ are monotonically non-decreasing.
\end{enumerate}
 (Note that $C^{(l)}$ is not necessarily equal to $A^{(l)}\star B^{(l)}$.) In the end when $l=0$ we can get the matrix $C^{(0)}_{i,j}=C_{i,j}\mod p,$ by the procedure of recursion. Thus we can calculate the exact value of $C_{i,j}$ by the result of $C_{i,j}\mod p.$

We can see all elements in $A^{(l)}, B^{(l)}, C^{(l)}$ are non-negative integers at most $O(n^\alpha/2^l)$ or infinite. From $B$ is row-monotone and property (2) of $C^{(l)}$, every row of $B^{(l)}, C^{(l)}$ composed of $O(n/2^l)$ intervals, where all elements in each interval are the same. Define a segment as:
\begin{definition}
A \emph{segment} $(i,k,[j_0,j_1])$ w.r.t. $B^{(l)}$ and  $C^{(l)}$, where $i,k,j_0,j_1\in [n]$ and $j_0\leq j_1$, satisfies that for all $j_0\leq j\leq j_1$, $B^{(l)}_{k,j}=B^{(l)}_{k,j_0}$, $B^*_{k,j}=B^*_{k,j_0}$ and $C^{(l)}_{i,j}=C^{(l)}_{i,j_0}$, $C^*_{i,j}=C^*_{i,j_0}$.
\end{definition}
Then each pair of rows of $B^{(l)}, C^{(l)}$ can be divided into $O(n/2^l)$ segments.

We maintain the auxiliary sets $T^{(l)}_b$ for $-10\le b\le 10$ throughout the algorithm, where the set $T^{(l)}_b$ consists of all the segments $(i,k,[j_0,j_1])$ w.r.t. $B^{(l)}$ and  $C^{(l)}$ satisfying: (So this holds for all $j\in[j_0,j_1]$.) $$\text{$A_{i,k}$ is finite and $A^*_{i,k}+B^*_{k,j_0}\neq C^*_{i,j_0}$ and $A^{(l)}_{i,k}+B^{(l)}_{k,j_0}= C^{(l)}_{i,j_0}+b$ }$$

The algorithm proceeds as:
\begin{itemize}
    \item In the first iteration $l=h$, we want to calculate $C^{(h)}_{i,j}.$ However since $p<2^h$, $A^{(h)},B^{(h)},C^{(h)}$ are zero matrices, so $T^{(h)}_0$ includes all segments $(i,k,[j_0,j_1])$ where $A_{i,k}$ is finite and $A^*_{i,k}+B^*_{k,j_0}\neq C^*_{i,j_0}$. And $T^{(h)}_b=\emptyset~(b\ne 0)$. Since the number of segments in a row w.r.t. $B^{(h)},C^{(h)}$ is $O(n^{1-\alpha})$, $|T^{(h)}_b|=O(n^{3-\alpha})$.
    \item For $l=h-1,\cdots, 0$, we first compute $C^{(l)}$ with the help of $T^{(l+1)}_b$, then construct $T^{(l)}_b$ from $T^{(l+1)}_b$. By Lemma~\ref{mono-triple} that $\bigcup_{i=-10}^{10}T^{(l)}_i\subseteq \bigcup_{i=-10}^{10}T^{(l+1)}_i$, we can search the shorter segments contained in $T^{(l+1)}_b$ to find $T^{(l)}_b$. By Lemma~\ref{mono-tb}, $|T^{(l)}_b|$ is always bounded by $O(n^{3-\alpha})$. 
\end{itemize}

Each iteration has three phases:

\noindent\textbf{Polynomial matrix multiplication.} Construct two polynomial matrices $A^p$ and $B^p$ on variables $x, y$ in the following way:
When $A_{i,k}$ is finite,
$$A^p_{i, k} = x^{A^{(l)}_{i, k} - 2A^{(l+1)}_{i, k}}\cdot y^{ A^{(l+1)}_{i, k}}$$
Otherwise $A^p_{i,k}=0$, and:
$$B^p_{k, j} = x^{B^{(l)}_{k, j} - 2B^{(l+1)}_{k, j}}\cdot y^{B^{(l+1)}_{k, j}}$$
Then, compute the standard $(+, \times)$ matrix multiplication $C^p = A^p\cdot B^p$ using fast matrix multiplication algorithms. Note that $A^{(l)}_{i, j}-2A^{(l+1)}_{i, j}, B^{(l)}_{i, j}-2B^{(l+1)}_{i, j}$ are $0$ or $1,$ so the degree of $x$ terms are 0 or 1. This phase runs in time $\tilde{O}(n^{\omega+\alpha}).$ 

\noindent\textbf{Subtracting erroneous terms.} 
This phase is to extract the true values $C^{(l)}_{i, j}$'s from $C^{(l+1)}_{i, j}.$ The algorithm iterates over all offsets $-10\le b\le 10$, and enumerates all the segments in $T_b^{(l+1)}.$

For each pair of indices $i, j\in [n]$, if $C^p_{i,j}=0$ then $C^{(l)}_{i,j}=+\infty$, otherwise collect all the monomials $\lambda x^cy^d$ of $C_{i, j}^p$ such that 
$$d= {C}^{(l+1)}_{i, j} + b$$ 
and let $C_{i, j, b}^p(x)$ be the sum of all such terms $\lambda x^c$. Next, compute a polynomial 
$$R_{i, j, b}^p(x) = \sum_{(i, k, [j_0,j_1])\in T^{(l+1)}_b, j\in [j_0,j_1]} x^{A^{(l)}_{i, k} - 2A^{(l+1)}_{i, k} + B^{(l)}_{k,j} - 2B^{(l+1)}_{k,j}}$$
Finally, let $s_{i, j, b}$ be the minimum degree of $x$ in the polynomial $C^p_{i, j, b}(x) - R^p_{i, j, b}(x)$, and compute a candidate value $c_{i, j, b} = 2d + s_{i, j, b} $. (If $s_{i,j,b}=0$ then $c_{i,j,b}=+\infty$.) Ranging over all integer offsets $-10\le b\le 10$, take the minimum of all candidate values and output as $C^{(l)}_{i, j} = \min_{-10\le b\le 10}\{c_{i, j, b}\}$. This phase runs in time $\tilde{O}(n^{3-\alpha}+n^{2+\alpha})$, since every segment $(i, k, [j_0,j_1])\in T^{(l+1)}_b$ contains at most two different $B^{(l)}_{k,j}$, thus also two different $R_{i, j, b}^p(x)$, so we can use a segment tree to compute all of $C^p_{i, j, b}(x) - R^p_{i, j, b}(x)$ in $\tilde{O}(n^{2+\alpha}+|T^{(l+1)}_b|)$ time.

\noindent\textbf{Computing Triples $T^{(l)}_b$.} 
Since $B^{(l)}_{k,j}-2B^{(l+1)}_{k,j}$ and $C^{(l)}_{i,j}-2C^{(l+1)}_{i,j}$ are both between 0 and a constant (see Lemma~\ref{relation-cl}), so each segment w.r.t. $B^{(l+1)}, C^{(l+1)}$ can be split into at most $O(1)$ segments w.r.t. $B^{(l)}, C^{(l)}$. 
%To compute $T^{(l)}_b$, initially set all $T^{(l)}_b\leftarrow \emptyset$ for all $|b|\leq 10$. 
By Lemma \ref{mono-triple} we know that $\bigcup_{i=-10}^{10}T^{(l)}_i$ is contained in $\bigcup_{i=-10}^{10}T^{(l+1)}_i,$ so our work here is to check the sub-segments of each segment in $\bigcup_{i=-10}^{10}T^{(l+1)}_i$ and put it into the $T^{(l)}_b$ it belongs to. Each segment in $T^{(l+1)}_b$ breaks into at most $O(1)$ sub-segments in the next iteration, and we can use binary search to find the breaking points. This phase runs in time $\tilde{O}(n^{3-\alpha}).$

The expected running time of the recursive algorithm is bounded by $\tilde{O}(n^{3-\alpha} + n^{\omega+\alpha})$. Taking $\alpha = (3 - \omega) / 2$, the running time becomes $\tilde{O}(n^{(3+\omega) /2})$.

\subsubsection{Proof of correctness}

We first prove the lemmas needed to bound the running time and show the correctness, then we will show that the properties of $C^{(l)}_{i,j}$ are maintained in the algorithm:

\begin{lemma}\label{relation-cl}
In each iteration $l=h-1,\cdots,0$, $-7\leq C^{(l)}_{i,j}-2C^{(l+1)}_{i,j}\leq 8$.
\end{lemma}
\begin{proof}
For all $l$, we can get:
$$ \frac{(C_{i,j} \mod p)}{ 2^{l}}-3  \leq C^{(l)}_{i, j} \leq  \frac{(C_{i,j}\mod p)}{ 2^{l}}+2$$
and 

$$ 2C^{(l+1)}_{i,j}-7\leq 2\frac{(C_{i,j} \mod p)}{ 2^{l+1}}-3\leq C^{(l)}_{i,j}\leq  2\frac{(C_{i,j}\mod p)}{ 2^{l+1}}+2 \leq 2C^{(l+1)}_{i,j}+8$$
\end{proof}

\begin{lemma}\label{mono-triple}
%$T^{(l)}_b \subseteq \bigcup_{i=\lfloor b/2-1\rfloor}^{\lfloor (b+1)/2\rfloor}T^{(l+1)}_i,$ for each $-10\le b\le 10.$ Therefore 
We have $\bigcup_{i=-10}^{10}T^{(l)}_i\subseteq \bigcup_{i=-10}^{10}T^{(l+1)}_i$, that is, the segments we consider in each iteration must be sub-segments of the segments in the last iteration.
    
\end{lemma}

\begin{proof}

Segments $(i,k,[j_0,j_1])$ in $T^{(l)}_b$ and $T^{(l+1)}_b$ must satisfy $A_{i,k}$ is finite and $A^*_{i,k}+B^*_{k,j_0}\neq C^*_{i,j_0}$. By definition, $A^{(l)}_{i,k}-2A^{(l+1)}_{i,k}=0$ or $1.$ Similar for $B$, and by Lemma~\ref{relation-cl}, we have

\begin{equation*}
\begin{split}
A^{(l+1)}_{i, k} +B^{(l+1)}_{k, j} -C^{(l+1)}_{i, j}
&\ge A^{(l)}_{i, k}/2 -1/2 +B^{(l)}_{k, j}/2 -1/2 -C^{(l)}_{i, j}/2 -7/2\\
&\ge \frac{1}{2}\left(A^{(l)}_{i, k} +B^{(l)}_{k, j}-C^{(l)}_{i, j}\right)-9/2.\\
A^{(l+1)}_{i, k} +B^{(l+1)}_{k, j} -C^{(l+1)}_{i, j}
&\le A^{(l)}_{i, k}/2 +B^{(l)}_{k, j}/2 -C^{(l)}_{i, j}/2 + 4\\
&\le \frac{1}{2}\left(A^{(l)}_{i, k} +B^{(l)}_{k, j}-C^{(l)}_{i, j}\right)+4.\\
\end{split}
\end{equation*}

Therefore, when $-10\le A^{(l)}_{i, k} +B^{(l)}_{k, j} -C^{(l)}_{i, j}\le 10,$
\begin{equation*}
    -10<-10/2-9/2\le A^{(l+1)}_{i, k} +B^{(l+1)}_{k, j} -C^{(l+1)}_{i, j}\le 10/2+4<10.
\end{equation*}

\end{proof}

\begin{lemma}\label{mono-tb}
The expected number of segments in $T^{(l)}_b$ is $\tilde{O}(n^{3-\alpha}).$
\end{lemma}

\begin{proof}

When $2^l>p/100$, the total number of segments is bounded by $O(n^{3-\alpha})$, so next we assume that $2^l<p/100$.

For any segment $(i,k,[j_0,j_1])$, and arbitrarily pick a $j\in[j_0,j_1]$ where $A_{i,k}$ is finite and $A^*_{i,k}+B^*_{k,j}\neq C^*_{i,j}$. By Assumption~\ref{basic-assumption}, $(C_{i,j}\mod p)<2p/3$, so if $A^*_{i,k}+B^*_{k,j}\geq C^*_{i,j}+1$, $$A_{i,k}/p+B_{k,j}/p\geq \lfloor A_{i,k}/p\rfloor+\lfloor B_{k,j}/p\rfloor \geq \lfloor C_{i,j}/p\rfloor+1 \geq C_{i,j}/p-2/3+1=C_{i,j}/p+1/3$$ So $A_{i,k}+B_{k,j}\geq C_{i,j}+p/3$. Similarly, if $A^*_{i,k}+B^*_{k,j}\leq C^*_{i,j}-1$, 
$$A_{i,k}/p-1/3+B_{k,j}/p-1/3\leq \lfloor A_{i,k}/p\rfloor+\lfloor B_{k,j}/p\rfloor \leq \lfloor C_{i,j}/p\rfloor-1 \leq C_{i,j}/p-1$$ Thus we get $|A_{i,k}+B_{k,j}- C_{i,j}|\geq p/3$ in either case.

We want to bound the probability that $(i,k,[j_0,j_1])$ appears in $T^{(l)}_{b}.$ By definition, this is to say that
\begin{equation*}
\left\lfloor \frac{A_{i,k} \mod p}{2^l} \right\rfloor +\left\lfloor \frac{B_{k,j} \mod p}{2^l}\right\rfloor=C^{(l)}_{i,j}+b.
\end{equation*}

So $$-4\le \frac{A_{i,k} \mod p}{2^l}+\frac{B_{k,j} \mod p}{2^l}-\frac{C_{i,j} \mod p}{2^l}-b\le 4$$ 
%Note that $(A_{i,k} \mod p)+(B_{k,j} \mod p)-(C_{i,j} \mod p)-((A_{i,k} +B_{k,j}-C_{i,j})\mod p)= 0$ or $\pm p.$  
Let $C_{i,j}=A_{i,q}+B_{q,j}$, and
\begin{equation*}
    (A_{i,k} +B_{k,j}-A_{i,q}-B_{q,j})\mod p\in[2^l (b-4),2^l (b+4)].
\end{equation*}

That is, $A_{i,k} +B_{k,j}-A_{i,q}-B_{q,j}$ should be congruent to one of the $O(2^l)$ remainders. For each possible remainder $r\in[2^l(b-4),2^l(b+4)],$ ($|b|\leq 10$), we have
\begin{equation*}
    |r|\le 14\cdot 2^l<p/6 \leq \frac{1}{2}\mid A_{i,k} +B_{k,j}-A_{i,q}-B_{q,j}\mid.
\end{equation*}

If $A_{i,k},A_{i,q}$ are finite and $B_{k,j},B_{q,j}$ are from the original $B$ (see Lemma~\ref{assumption-correctness}),  $|(A_{i,k} +B_{k,j}-A_{i,q}-B_{q,j})-r|$ is a positive number bounded by $O(n)$, the number of different primes $p\in[40n^\alpha,80n^\alpha]$ which divides $(A_{i,k} +B_{k,j}-A_{i,q}-B_{q,j})-r$ can not exceed $1/\alpha=O(1).$ In our algorithm, when we uniformly choose a prime $p$ from $[40n^\alpha,80n^\alpha],$ the probability that $(A_{i,k} +B_{k,j}-A_{i,q}-B_{q,j})\mod p \equiv r$ is $\tilde{O}\left(\frac{1}{n^\alpha}\right).$ 

However in Lemma~\ref{assumption-correctness}, $B_{k,j}$ and $B_{q,j}$ may be set artificially to numbers which are congruent to 0, $\lceil p/3\rceil$ or $\lceil 2p/3\rceil$ modulo $p$, but finite $A_{i,k}$ and $A_{i,q}$ must come from the original $A$. For example, if $B_{k,j}$ is made congruent to $\lceil p/3\rceil$ module $p$ and $B_{q,j}$ is from original $B$, we want that $p$ divides $A_{i,k} -A_{i,q}-B_{q,j}-r+\lceil p/3\rceil$. Since 3 does not divides $p$, $3\lceil p/3\rceil$ is $p+1$ or $p+2$, so $p$ divides $3(A_{i,k} -A_{i,q}-B_{q,j}-r)+1$ or $3(A_{i,k} -A_{i,q}-B_{q,j}-r)+2$. The probability is still $\tilde{O}\left(\frac{1}{n^\alpha}\right)$. Other cases of $B_{k,j}$ and $B_{q,j}$ can be done similarly. Since on all cases of $B_{k,j}$ and $B_{q,j}$ the conditional probability that $p$ divides $(A_{i,k} +B_{k,j}-A_{i,q}-B_{q,j})-r$ is bounded by $\tilde{O}\left(\frac{1}{n^\alpha}\right)$, the total probability is also $\tilde{O}\left(\frac{1}{n^\alpha}\right)$.

Since there are $O(2^l)$ such possible remainders $r$, in expectation we have $O(2^l)\cdot O\left(\frac{n^{3}}{2^l}\right)\cdot\tilde{O}\left(\frac{1}{n^\alpha}\right)=\tilde{O}(n^{3-\alpha})$ segments in $T^{(l)}_b.$

\end{proof}

\begin{lemma}\label{exact-approx}
If $A_{i,k}+B_{k,j}=C_{i,j},$ then $A^{(l)}_{i,k}+B^{(l)}_{k,j}=C^{(l)}_{i,j}+b$ for some $-10\le b\le 10.$
\end{lemma}

\begin{proof}
By Assumption~\ref{basic-assumption},
\begin{equation*}
\begin{split}
A^{(l)}_{i,k}+B^{(l)}_{k,j}-C^{(l)}_{i,j}&=\left\lfloor \frac{A_{i,k}\mod p}{2^l}\right\rfloor+\left\lfloor \frac{B_{k,j}\mod p}{2^l}\right\rfloor-C^{(l)}_{i,j}\\
&\le  \frac{A_{i,k}\mod p}{2^l}+ \frac{B_{k,j}\mod p}{2^l}- \frac{C_{i,j}\mod p}{2^l}+3\\
&=\frac{(A_{i,k}+B_{k,j}-C_{i,j})\mod p}{2^l}+3=3.\\
A^{(l)}_{i,k}+B^{(l)}_{k,j}-C^{(l)}_{i,j}&=\left\lfloor \frac{A_{i,k}\mod p}{2^l}\right\rfloor+\left\lfloor \frac{B_{k,j}\mod p}{2^l}\right\rfloor-C^{(l)}_{i,j}\\
&\ge  \frac{A_{i,k}\mod p}{2^l}+ \frac{B_{k,j}\mod p}{2^l}- \frac{C_{i,j}\mod p}{2^l}-4\\
&=\frac{(A_{i,k}+B_{k,j}-C_{i,j})\mod p}{2^l}-4=-4.\\
\end{split}
\end{equation*}

\end{proof}

Next we argue that our algorithm correctly computes all entries $C^{(l)}_{i, j}$ from $C^{(l+1)}_{i, j}$ and $T^{(l+1)}_b$, for $l=h-1,\cdots, 0$. Let $q$ be the index such that $C_{i, j} = A_{i, q} + B_{q, j}$. By the above lemma, there exists an integer offset $b\in [-10,10]$ such that $A^{(l+1)}_{i, q} + B^{(l+1)}_{q, j} = C^{(l+1)}_{i, j} + b$. Therefore, by construction of polynomial matrices $A^p, B^p$, we have:
$$\begin{aligned}
	C^p_{i, j, b}(x) &= \sum_{k\mid A^{(l+1)}_{i, k} + B^{(l+1)}_{k, j} = C^{(l+1)}_{i, j} + b} x^{A^{(l)}_{i, k} - 2 A^{(l+1)}_{i, k} + B^{(l)}_{k, j} - 2 B^{(l+1)}_{k, j}} \\
    &= \sum_{k\mid (A^*_{i,k}+B^*_{k,j}=C^*_{i,j}) \land (A^{(l+1)}_{i, k} + B^{(l+1)}_{k, j} = C^{(l+1)}_{i, j} + b)} x^{A^{(l)}_{i, k} - 2 A^{(l+1)}_{i, k} + B^{(l)}_{k, j} - 2 B^{(l+1)}_{k, j}}\\
    &+    \sum_{k\mid (A^*_{i,k}+B^*_{k,j}\neq C^*_{i,j}) \land (A^{(l+1)}_{i, k} + B^{(l+1)}_{k, j} = C^{(l+1)}_{i, j} + b)} x^{A^{(l)}_{i, k} - 2 A^{(l+1)}_{i, k} + B^{(l)}_{k, j} - 2 B^{(l+1)}_{k, j}} \\
    &= \sum_{k\mid (A^*_{i,k}+B^*_{k,j}=C^*_{i,j}) \land (A^{(l+1)}_{i, k} + B^{(l+1)}_{k, j} = C^{(l+1)}_{i, j} + b)} x^{A^{(l)}_{i, k} - 2 A^{(l+1)}_{i, k} + B^{(l)}_{k, j} - 2 B^{(l+1)}_{k, j}} \\
    &+ \sum_{(i,k,[j_0,j_1])\in T^{(l+1)}_b, j\in [j_0,j_1]} x^{A^{(l)}_{i, k} - 2 A^{(l+1)}_{i, k} + B^{(l)}_{k, j} - 2 B^{(l+1)}_{k, j}} \\
    &= x^{-2(C^{(l+1)}_{i,j}+b)}\cdot\sum_{k\mid (A^*_{i,k}+B^*_{k,j}=C^*_{i,j}) \land (A^{(l+1)}_{i, k} + B^{(l+1)}_{k, j} = C^{(l+1)}_{i, j} + b)} x^{A^{(l)}_{i, k} + B^{(l)}_{k, j}} + R^p_{i,j,b}(x) 
\end{aligned}$$
Since $A^*_{i,q}+B^*_{q,j}=C^*_{i,j}$ and $A^{(l+1)}_{i, q} + B^{(l+1)}_{q, j} = C^{(l+1)}_{i, j} + b$, when we extract $A^{(l)}_{i, q} + B^{(l)}_{q, j}$ from terms of $C_{i, j, b}^p(x) - R_{i, j, b}^p(x)$, it satisfies
$$\left\lfloor \frac{A_{i,q}\mod p}{2^l}\right\rfloor+\left\lfloor \frac{B_{q,j}\mod p}{2^l}\right\rfloor \leq \frac{(A_{i,q}+B_{q,j})\mod p}{2^l} =
\frac{C_{i,j}\mod p}{2^l}\leq \left\lfloor\frac{(C_{i,j}\mod p) +2^l-1}{2^l}\right\rfloor$$
$$\left\lfloor \frac{A_{i,q}\mod p}{2^l}\right\rfloor+\left\lfloor \frac{B_{q,j}\mod p}{2^l}\right\rfloor \geq \frac{((A_{i,q}+B_{q,j})\mod p)-2(2^l-1)}{2^l} \geq \left\lfloor\frac{(C_{i,j}\mod p)-2(2^l-1)}{2^l}\right\rfloor$$

Thus the term which gives $A^{(l)}_{i, q} + B^{(l)}_{q, j}$ can give a valid $C^{(l)}_{i,j}$. Also for every term which gives $A^{(l)}_{i, k} + B^{(l)}_{k, j}$ satisfying $A^*_{i,k}+B^*_{k,j}=C^*_{i,j}$ and $A_{i, k} +B_{k,j}\geq C_{i.j}$, 
$$\left\lfloor \frac{A_{i,k}\mod p}{2^l}\right\rfloor+\left\lfloor \frac{B_{k,j}\mod p}{2^l}\right\rfloor \geq \frac{((A_{i,k}+B_{k,j})\mod p)-2(2^l-1)}{2^l} \geq \left\lfloor\frac{(C_{i,j}\mod p)-2(2^l-1)}{2^l}\right\rfloor$$
So by choosing the minimum, we can get a valid $C^{(l)}_{i,j}$ satisfying property (1).

%To see that $C^{(l)}$ satisfies property (2), consider $C^*_{i,j_0}=C^*_{i,j_1}$ where $j_0<j_1$. For all the $k$ such that $A^*_{i,k}+B^*_{k,j_1}=C^*_{i,j_1}$, then $C^*_{i,j_0}\leq A^*_{i,k}+B^*_{k,j_0}\leq A^*_{i,k}+B^*_{k,j_1}=C^*_{i,j_1}$, so $B^*_{k,j_0}=B^*_{k,j_1}$ and $B^{(l)}_{k,j_0}\leq B^{(l)}_{k,j_1}$. Thus, for term $x^{A^{(l)}_{i,k}+B^{(l)}_{k,j_1}}$ which gives $C^{(l)}_{i,j_1}$, the term with $x^{A^{(l)}_{i,k}+B^{(l)}_{k,j_0}}$ also exist in $C^p_{i,j_0}$ and cannot be subtracted since $A^*_{i,k}+B^*_{k,j_0}=C^*_{i,j_0}$. If this result is not included in $C^p_{i,j_0,b}(x)-R^p_{i,j_0,b}(x)$ for all $-10\leq b\leq 10$, $C^{(l+1)}_{i,j_0}<C^{(l+1)}_{i,j_1}$ and $C^{(l)}_{i,j_0}$ must get a smaller result then $A^{(l)}_{i,k}+B^{(l)}_{k,j_1}$. This proves property (2).

To see that $C^{(l)}$ satisfies property (2), consider $C^*_{i,j_0}=C^*_{i,j_1}$ where $j_0<j_1$. For all the $k$ such that $A^*_{i,k}+B^*_{k,j_1}=C^*_{i,j_1}$, $C^*_{i,j_0}\leq A^*_{i,k}+B^*_{k,j_0}\leq A^*_{i,k}+B^*_{k,j_1}=C^*_{i,j_1}$, so $B^*_{k,j_0}=B^*_{k,j_1}$ and $B^{(l)}_{k,j_0}\leq B^{(l)}_{k,j_1}$. Thus, for term $A^{(l)}_{i,k}+B^{(l)}_{k,j_1}$ which gives $C^{(l)}_{i,j_1}$, the term with $A^{(l)}_{i,k}+B^{(l)}_{k,j_0}$ also exist in $C^p_{i,j_0}$ and cannot be subtracted since $A^*_{i,k}+B^*_{k,j_0}=C^*_{i,j_0}$. If this result is not included in $C^p_{i,j_0,b}(x)-R^p_{i,j_0,b}(x)$ for all $-10\leq b\leq 10$,
$A^{(l+1)}_{i,k}+B^{(l+1)}_{k,j_0}-C^{(l+1)}_{i,j_0}$ is larger than 10 or less than $-10$. If $A^{(l+1)}_{i,k}+B^{(l+1)}_{k,j_0}-C^{(l+1)}_{i,j_0}<-10$, 
\begin{equation*}
\begin{split}
-10 >& A^{(l+1)}_{i,k}+B^{(l+1)}_{k,j_0}-C^{(l+1)}_{i,j_0} \\
= &\left\lfloor \frac{A_{i,k}\mod p}{2^{l+1}}\right\rfloor+\left\lfloor \frac{B_{k,j_0}\mod p}{2^{l+1}}\right\rfloor-C^{(l+1)}_{i,j_0}\\
\ge & \frac{A_{i,k}\mod p}{2^{l+1}}+ \frac{B_{k,j_0}\mod p}{2^{l+1}}- \frac{C_{i,j_0}\mod p}{2^{l+1}}-4
\end{split}
\end{equation*}
So $(A_{i,k}\mod p)+(B_{k,j_0}\mod p)-(C_{i,j_0}\mod p)<0$, which is impossible since $A^*_{i,k}+B^*_{k,j_0}=C^*_{i,j_0}$. Thus, it can only be that $A^{(l+1)}_{i,k}+B^{(l+1)}_{k,j_0}-C^{(l+1)}_{i,j_0}>10$, so by inductive assumption and Lemma~\ref{relation-cl}, 
$$
C^{(l)}_{i,j_1} =A^{(l)}_{i,k}+B^{(l)}_{k,j_1}
\geq 2\left(A^{(l+1)}_{i,k}+B^{(l+1)}_{k,j_1}\right)
\geq 2\left(A^{(l+1)}_{i,k}+B^{(l+1)}_{k,j_0}\right)
 > 2C^{(l+1)}_{i,j_0}+20
\geq C^{(l)}_{i,j_0}+12
$$
This proves property (2).

\section{Monotone min-plus convolution}
\subsection{Basic Algorithm}
In this section we prove Theorem~\ref{conv} following the same algorithmic framework of Theorem~\ref{matrix}. The min-plus convolution $C = A\diamond B$ of two array $A$ and $B$ can be defined as $C_{k} = \min_{i=1}^{k-1}\{A_i + B_{k-i}\}, \forall 2\leq k\leq 2n$.
Take two constant parameters $\alpha,\beta \in (0, 1)$ which are to be determined in the end; for convenience let us assume $n^\alpha$ is an integer.\\

\noindent\textbf{Approximation.} Define two integral arrays  $\tilde{A}, \tilde{B}$ such that $\tilde{A}_{i} = \lfloor A_{i} / n^\alpha\rfloor$, $\tilde{B}_{i} = \lfloor B_{i} / n^\alpha\rfloor$. Therefore, $\tilde{A}, \tilde{B}$ is an integer array whose entries are bounded by $O(n^{1-\alpha})$, and both $\tilde{A}, \tilde{B}$ are  non-decreasing.

Next, compute the approximate min-plus convolution $\tilde{C} = \tilde{A}\diamond\tilde{B}$ combinatorially. Initialize each entry of $\tilde{C}$ to be $\infty$, and maintain $\tilde{C}$ using a segment tree that supports interval updates. Divide $A$ and $B$ into at most $O(n^{1-\alpha})$ consecutive intervals:
$$[n] = [1, a_2-1]\cup [a_2, a_3-1]\cup\cdots \cup[a_g, n]=[1, b_2-1]\cup [b_2, b_3-1]\cup\cdots\cup[b_h, n]$$
such that for each $i\in [a_l, a_{l+1}-1]$, $\tilde{A}_i$'s are all equal, and for each $j\in [b_k, b_{k+1}-1]$, $\tilde{B}_j$'s are all equal (assume $a_1 = b_1 = 1$ and $a_{g+1}=b_{h+1}=n+1$). Then, to compute $\tilde{C}$, take any pair of indices $k, l\in [g]\times[h]$, and update $\tilde{C}_i\leftarrow \min\{\tilde{C}_i, \tilde{A}_{a_l} + \tilde{B}_{b_k} \}$ for each index $a_l+b_k\leq i\leq a_{l+1} + b_{k+1}-2$ using the segment tree data structure maintained on array $\tilde{C}$. The total time is $\tilde{O}(n^{2-2\alpha})$.\\

\noindent\textbf{Polynomial multiplication.} Uniformly sample a random prime number $p$ in the range $[n^\beta, 2n^\beta]$. Construct two polynomial $A^p$ and $B^p$ on variables $x, y, z$ in the following way:
$$A^p(x, y, z) = \sum_{i=1}^{n} x^{A_{i} - n^\alpha\tilde{A}_{i}}\cdot y^{\tilde{A}_{i}\mod p}\cdot z^i$$
$$B^p(x, y, z)= \sum_{i=1}^{n} x^{B_{i} - n^\alpha\tilde{B}_{i}}\cdot y^{\tilde{B}_{i}\mod p}\cdot z^i$$
Then, compute polynomial multiplication $C^p(x, y, z) = A^p(x, y, z)\cdot B^p(x, y, z)$ using standard fast Fourier transform algorithms~\cite{Schnhage2005}.\\

\noindent\textbf{Subtracting erroneous terms.} The last phase is to extract the true values $C_i$'s from $\tilde{C}$ and $C^p(x, y, z)$. The algorithm iterates over all offsets $b\in \{0, 1, 2\}$, and computes the set $T_b\subseteq [n]^2$ of all pairs of indices $(i, j)$ such that $\tilde{A}_{i} + \tilde{B}_{j} \neq \tilde{C}_{i+j}+b$ but $\tilde{A}_{i} + \tilde{B}_{j} \equiv \tilde{C}_{i+j}+b\mod p$; in the running time analysis, we will show that $T_b$ can be computed in time $\tilde{O}(|T_b| + n^{2-2\alpha})$.

For each index $1\leq k\leq n$, consider the coefficient of $z^k$ in $C^p(x, y, z)$ denoted by $C_k^p(x, y)$. Enumerate all terms $\lambda x^cy^d$ of $C_k^p(x, y)$ such that $d \equiv \tilde{C}_k+b\mod p$, and define $C_{k, b}^p(x)$ to be the sum of all $\lambda x^c$. Next, compute a polynomial:
$$R_{k, b}^p(x) = \sum_{(i, k-i)\in T_b} x^{A_{i} - n^\alpha\tilde{A}_{i} + B_{k-i} - n^\alpha\tilde{B}_{k-i}}$$
Finally, let $s_{k, b}$ be the minimum degree of $x$ of the polynomial $C^p_{k, b}(x) - R^p_{k, b}(x)$, and compute a candidate value $n^\alpha(\tilde{C}_k+b) + s_{k, b}$ for $C_k$. Ranging over all integer offsets $b\in \{0, 1, 2\}$, take the minimum of all candidate values and output as $C_k = \min_{0\leq b\leq 2}\{n^\alpha(\tilde{C}_k+b) + s_{k, b} \}$.

\subsubsection{Proof of correctness}
\begin{lemma}
	For any pair of indices $1\leq i, j\leq n$ such that $A_{i} + B_{j} = C_{i+j}$, we have:
	$$0\leq \tilde{A}_{i} + \tilde{B}_{j} - \tilde{C}_{i+j}\leq 2$$
\end{lemma}
\begin{proof}
	Clearly $\tilde{A}_{i} + \tilde{B}_{j} - \tilde{C}_{i+j}\geq 0$ by definition of min-plus convolution. So let us only focus on the second inequality. Suppose $\tilde{C}_{i+j} = \tilde{A}_k + \tilde{B}_l$ for some indices $1\leq k, l\leq n$ such that $k+l = i+j$. Then, by definition of $\tilde{A}, \tilde{B}$, we have:
	$$\begin{aligned}
		n^\alpha\tilde{C}_{i+j} &= n^\alpha\tilde{A}_k + n^\alpha\tilde{B}_l \geq A_k + B_l - 2n^\alpha \geq C_{i+j} - 2n^\alpha\\
		& = A_i + B_j - 2n^\alpha\geq n^\alpha \tilde{A}_i + n^\alpha\tilde{B}_j - 2n^\alpha
	\end{aligned}$$
	Hence, $\tilde{A}_{i} + \tilde{B}_{j} - \tilde{C}_{i+j}\leq 2$.
\end{proof}

Next we argue that our algorithm correctly computes all entries $C_{k}$. Let $l$ be the index such that $C_{k} = A_{l} + B_{k-l}$. By the above lemma, there exists an integer offset $b\in [0, 2]$ such that $\tilde{A}_{l} + \tilde{B}_{k-l} = \tilde{C}_{k} + b$. Therefore, by construction of polynomial matrices $A^p, B^p$, we have:
$$\begin{aligned}
C^p_{k, b}(x) &= \sum_{i\mid \tilde{A}_{i} + \tilde{B}_{k-i} = \tilde{C}_{k} + b}x^{A_{i} - n^\alpha \tilde{A}_{i} + B_{k-i} - n^\alpha \tilde{B}_{k-i}}\\
&+\sum_{i\mid (\tilde{A}_{i} + \tilde{B}_{k-i} \neq \tilde{C}_{k} + b)\wedge (\tilde{A}_{i} + \tilde{B}_{k-i} \equiv \tilde{C}_{k} + b\mod p)}x^{A_{i} - n^\alpha \tilde{A}_{i} + B_{k-i} - n^\alpha \tilde{B}_{k-i}}\\
&= \sum_{k\mid \tilde{A}_{i} + \tilde{B}_{k-i} = \tilde{C}_{k} + b}x^{A_{i} - n^\alpha \tilde{A}_{i} + B_{k-i} - n^\alpha \tilde{B}_{k-i}} +\sum_{(i, k-i)\in T_b}x^{A_{i} - n^\alpha \tilde{A}_{i} + B_{k-i} - n^\alpha \tilde{B}_{k-i}}\\
&= x^{-n^\alpha (\tilde{C}_{k} + b)} \cdot\sum_{k\mid \tilde{A}_{i} + \tilde{B}_{k-i} =\tilde{C}_{k} + b}x^{A_{i}+ B_{k-i}} + R^p_{k, b}(x)
\end{aligned}$$
Therefore, $x^{-n^\alpha (\tilde{C}_{k} + b)} \cdot\sum_{i\mid \tilde{A}_{i} + \tilde{B}_{k-i} =\tilde{C}_{k} + b}x^{A_{i}+ B_{k-i}}= C_{k, b}^p(x) - R_{k, b}^p(x)$. Since $\tilde{A}_{l} + \tilde{B}_{k-l} =\tilde{C}_{k} + b$, we know $n^\alpha(\tilde{C}_k + b) + s_{k, b} = C_k$.

\subsubsection{Running time analysis}
By the algorithm, computing the approximation array $\tilde{C}$ takes time $\tilde{O}(n^{2-2\alpha})$. As for polynomial multiplication, by definition the $x$-degree and $y$-degree of $A^p, B^p$ are both bounded in absolute value by $O(n^\alpha), O(n^\beta)$ respectively, so the polynomial multiplication takes time $O(n^{1 + \alpha + \beta})$.
\begin{lemma}
	The set $T_b$ can be computed in time $\tilde{O}(|T_b| + n^{2-2\alpha})$.
\end{lemma}
\begin{proof}
	Recall the partition of sequences $A, B$ into intervals in the first phase:
	$$[n] = [1, a_2-1]\cup [a_2, a_3-1]\cup\cdots \cup[a_g, n]=[1, b_2-1]\cup [b_2, b_3-1]\cup\cdots\cup[b_h, n]$$
	such that $A, B$ are all equal within each interval. Fix two intervals $[a_s, a_{s+1}-1]$ and $[b_t, b_{t+1}-1]$ where array $\tilde{A}$ and $\tilde{B}$ have the same value. Next, we try to find all $i\in [a_s, a_{s+1}-1], j\in [b_t, b_{t+1}-1]$ such that $(i, j)\in T_b$. The key advantage is that the value $\Delta = \tilde{A}_i + \tilde{B}_j$ is a fixed value for any $(i,j)\in [a_s, a_{s+1}-1]\times [b_t, b_{t+1}-1]$.
	
	Now search for all indices $a_s+b_t\leq k\leq a_{s+1}+b_{t+1}-2$ such that $\tilde{C}_k+b\neq\Delta$ while $\tilde{C}_k + b\equiv \Delta\mod p$. Using standard binary search tree data structures, we can obtain the set of all such indices $K_b$ in time $\tilde{O}(|K_b|)$. Then, for each $k\in K_b$, enumerate all $(i, k-i)$ satisfying 
	$$\max\{a_s, k+1-b_{t+1} \}\leq i\leq \min\{a_{s+1}-1, k-b_t\}$$ 
	and add the pair $(i, k-i)$ to $T_b$. Ranging over all choices of intervals $[a_s, a_{s+1}-1]$ and $[b_t, b_{t+1}-1]$, the total time becomes $\tilde{O}(|T_b| + n^{2-2\alpha})$.
\end{proof}

By the above lemma, the subtraction phase takes time $\tilde{O}(|T_b| + n^{2-2\alpha})$ as well. So it suffices to bound the size of $T_b$. For any $(i, k-i)$ such that $\tilde{A}_{i} + \tilde{B}_{k-i}\neq \tilde{C}_k + b$, since $p$ is a uniformly random prime in the range $[n^\beta, 2n^\beta]$, the probability that $\tilde{A}_{i} + \tilde{B}_{k-i} - \tilde{C}_k - b$ can be divided by $p$ is bounded by $\tilde{O}(n^{-\beta})$. Hence, by linearity of expectation, $\mathbb{E}_p[|T_b|]\leq \tilde{O}(n^{2-\beta})$.

Throughout all three phases, the expected running time of our algorithm is bounded by $\tilde{O}(n^{2-2\alpha} + n^{1+\alpha+\beta} + n^{2-\beta})$. Taking $\alpha = 0.2, \beta =0.4$, the running time becomes $\tilde{O}(n^{1.6})$.

\subsection{Recursive Algorithm}

Let $\alpha\in (0, 1)$ be a constant parameter to be determined later, and pick a uniformly random prime number $p$ in the range of $[40n^\alpha, 80n^\alpha]$. Without loss of generality, let us assume that $n$ is a power of $2$. Like in Section~\ref{sec:recursive}, w.l.o.g. we make the following assumption about elements in $A$ and $B$:

\begin{assumption}\label{conv-basic-assumption}
For every $i$, either $(A_{i}\mod p)<p/3$ or $A_{i}=+\infty$, and $A$ is monotone besides the infinite elements. Similar for $B$.
\end{assumption}

\begin{lemma}\label{conv-assumption}
The general computation of $A\diamond B$ can be reduced to a constant number of computations of $A^i\diamond B^i$ where all of $A^i,B^i$'s satisfy Assumption~\ref{conv-basic-assumption}.
The number of intervals of infinity in each $A^i$ and $B^i$ is bounded by $O(n^{1-\alpha})$.
\end{lemma}
\begin{proof}
We just arrange the elements of $A$ to $A',A'',A'''$ by their remainders module $p$, other elements becomes $+\infty$. It is easy to see that the number of intervals of infinity in each of $A',A'',A'''$ is bounded by $O(n^{1-\alpha})$. Similar for $B$.
\end{proof}

Define integer $h$ such that $2^{h-1}\leq p< 2^h$. For each integer $0\leq l\leq h$, let $A^{(l)}$ be a sequence of length $n$ defined as $A^{(l)}_{i} = \lfloor \frac{A_{i}\mod p}{ 2^{l}} \rfloor$ if $A_{i}$ is finite, otherwise $A^{(l)}_{i}=+\infty$, similarly define sequence $B^{(l)}= \lfloor \frac{B_{i}\mod p}{ 2^{l}} \rfloor$. 

We will recursively calculate $C^{(l)}$ for $l=h,h-1,\cdots,0$, and if $C_i$ is finite, $C^{(l)}$ will satisfy 
$$\lfloor \frac{(C_{i} \mod p)-2(2^l-1)}{ 2^{l}} \rfloor \leq C^{(l)}_{i} \leq \lfloor \frac{(C_{i}\mod p)+2(2^l-1)}{ 2^{l}} \rfloor$$  (Note that $C^{(l)}$ is not necessarily equal to $A^{(l)}\diamond B^{(l)}$.) In the end when $l=0$ we can get the matrix $C^{(0)}_{i}=C_{i}\mod p$ by the procedure of recursion. Define $A^*$ and $B^*$ as $A^*_{i}=\lfloor A_{i}/p\rfloor$ and $B^*_{i}=\lfloor B_{i}/p\rfloor$. We use the segment tree structure to calculate $C^*=A^*\diamond B^*$ in $\tilde{O}(n^{2-2\alpha})$ time. By Assumption~\ref{conv-basic-assumption}, $\tilde{C}_{i}=\lfloor C_{i}/p\rfloor$ if $C_{i}$ is finite. Thus we can calculate the exact value of $C_{i}$ by the result of $C_{i}\mod p.$

%We also define $C^{(l)}$ as $C^{(l)}_{i} = \lfloor \frac{C_{i}\mod p}{ 2^{l}} \rfloor$ if $C_{i}$ is finite, and we will recursively calculate $C^{(l)},$ for $l=h,h-1,\cdots,0.$ Note that $C^{(l)}$ is not necessarily equal to $A^{(l)}\star B^{(l)}$. In the end we can get the sequence $C^{(0)}_{i}=C_{i}\mod p,$ by the procedure of recursion. We let $\tilde{A}$ and $\tilde{B}$ as $\tilde{A}_{i}=\lfloor A_{i}/p\rfloor$ and $\tilde{B}_{i}=\lfloor B_{i}/p\rfloor$. We can apply the min-plus product algorithm for small integers [] to calculate a $\tilde{C}=\tilde{A}\star \tilde{B}$ in $\tilde{O}(n^{2-2\alpha})$ time. By Assumption~\ref{conv-basic-assumption}, $\tilde{C}_{i}=\lfloor C_{i}/p\rfloor$ if $C_{i}$ is finite. Thus we can calculate the exact value of $C_{i}$ by the result of $C_{i}\mod p.$

We can see all elements in $A^{(l)}, B^{(l)}, C^{(l)}$ are non-negative integers at most $O(n^\alpha/2^l)$ or infinite. Since $A,B$ are monotone and by Lemma~\ref{conv-assumption}, $A^{(l)}, B^{(l)}$ compose of $O(n/2^l)$ intervals, where all elements in each interval are the same. Define a segment as:
\begin{definition}
A \emph{segment} $([i_0,i_1],k)$ w.r.t. $A^{(l)}$ and  $B^{(l)}$, where $i_0,i_1,k\in [n]$ and $i_0\leq i_1$, satisfies that for all $i_0\leq i\leq i_1$, $A_{i},B_{k-i}$ are finite, and $A^{(l)}_{i}=A^{(l)}_{i_0}$, $A^*_{i}=A^*_{i_0}$ and $B^{(l)}_{k-i}=B^{(l)}_{k-i_0}$, $B^*_{k-i}=B^*_{k-i_0}$.
\end{definition}
Then $A^{(l)}, B^{(l)}$ can be divided into $O(n/2^l)$ segments for some $k$.

We maintain the auxiliary sets $T^{(l)}_b$ for $-10\le b\le 10$ throughout the algorithm, where the set $T^{(l)}_b$ consists of all the segments $([i_0,i_1],k)$ w.r.t. $A^{(l)}$ and  $B^{(l)}$ satisfying: $$\text{$C_{k}$ is finite and $A^*_{i_0}+B^*_{k-i_0}\neq C^*_{k}$ and $A^{(l)}_{i_0}+B^{(l)}_{k-i_0}= C^{(l)}_{k}+b$ }$$

The algorithm proceeds as:
\begin{itemize}
    \item In the first iteration $l=h$, we want to calculate $C^{(h)}.$ However since $p<2^h$, $A^{(h)},B^{(h)},C^{(h)}$ are zero sequences, so $T^{(h)}_0$ includes all segments $([i_0,i_1],k)$ where $A^*_{i_0}+B^*_{k-i_0}\neq C^*_{k}$, and $T^{(h)}_b=\emptyset~(b\ne 0)$. Since the number of segments w.r.t. $A^{(h)},B^{(h)}$ for every $k$ is $O(n^{1-\alpha})$, $|T^{(h)}_b|=O(n^{2-\alpha})$.
    \item For $l=h-1,\cdots, 0$, we first compute $C^{(l)}$ with the help of $T^{(l+1)}_b$, then construct $T^{(l)}_b$ from $T^{(l+1)}_b$. By Lemma~\ref{conv-triple} that $\bigcup_{i=-10}^{10}T^{(l)}_i\subseteq \bigcup_{i=-10}^{10}T^{(l+1)}_i$, we can search the shorter segments contained in $T^{(l+1)}_b$ to find $T^{(l)}_b$. By Lemma~\ref{conv-tb}, $|T^{(l)}_b|$ is always bounded by $O(n^{2-\alpha})$. Each iteration has three phases:
\end{itemize}

\noindent\textbf{Polynomial multiplication.} Construct two polynomial matrices $A^p$ and $B^p$ on variables $x, y$ in the following way:
$$A^p= \sum_{i=1}^n x^{A^{(l)}_{i} - 2A^{(l+1)}_{i}}\cdot y^{ A^{(l+1)}_{i}}\cdot z^i.$$
$$B^p = \sum_{j=1}^n x^{B^{(l)}_{j} - 2B^{(l+1)}_{j}}\cdot y^{B^{(l+1)}_{j}}\cdot z^j.$$
Then, compute the polynomial multiplication $C^p = A^p\cdot B^p$ using standard FFT~\cite{Schnhage2005}. Note that $A^{(l)}_{i}-2A^{(l+1)}_{i}, B^{(l)}_{j}-2B^{(l+1)}_{j}$ are $0$ or $1,$ so the degree of $x$ terms are 0 or 1. This phase runs in time $\tilde{O}(n^{1+\alpha}).$

\noindent\textbf{Subtracting erroneous terms.} 

This phase is to extract the true values $C^{(l)}_{k}$'s from $C^{(l+1)}_{k}.$ The algorithm iterates over all offsets $-10\le b\le 10$, and enumerates all the segments in $T_b^{(l+1)}.$

For each index $1\leq k\leq n$, consider the coefficient of $z^k$ in $C^p$ denoted by $C_k^p(x, y)$. Enumerate all terms $\lambda x^cy^d$ of $C_k^p(x, y)$ such that $d= {C}^{(l+1)}_{k} + b$, and define $C_{k, b}^p(x)$ to be the sum of all such $\lambda x^c$. Next, compute a polynomial:
$$R_{k, b}^p(x) = \sum_{([i_0,i_1],k)\in T_b^{(l+1)}, i\in[i_0,i_1]} x^{A^{(l)}_{i} - 2A^{(l+1)}_{i} + B^{(l)}_{k-i} - 2B^{(l+1)}_{k-i}}$$
Finally, let $s_{k, b}$ be the minimum degree of $x$ of the polynomial $C^p_{k, b}(x) - R^p_{k, b}(x)$, and compute a candidate value $s_{k, b} + 2d$ for $C_k^{(l)}$. Ranging over all integer offsets $-10\le b\le 10$, take the minimum of all candidate values and output as $C_k^{(l)} = \min_{-10\le b\le 10}\{s_{k, b} + 2d \}$.

\noindent\textbf{Computing Triples $T^{(l)}_b$.} 

To compute $T^{(l)}_b$, initially set all $T^{(l)}_b\leftarrow \emptyset$ for all $|b|\leq 10$. By Lemma \ref{conv-triple} we know that $\bigcup_{i=-10}^{10}T^{(l)}_i$ is contained in $\bigcup_{i=-10}^{10}T^{(l+1)}_i,$ so our work here is to check each segment in $\bigcup_{i=-10}^{10}T^{(l+1)}_i$ and put it into the $T^{(l)}_b$ it belongs to. Each segment in $T^{(l+1)}_b$ breaks into at most 4 segments in the next iteration, and we can use binary search to find the breaking points. This phase runs in time $\tilde{O}(n^{2-\alpha})$ by Lemma~\ref{conv-tb}.

The expected running time of the recursive algorithm is bounded by $\tilde{O}(n^{1+\alpha}+ n^{2-\alpha})$. Taking $\alpha = 0.5$, the running time becomes $\tilde{O}(n^{1.5})$.

\subsection{Proof of correctness}

\begin{lemma}\label{conv-triple}
We have $\bigcup_{i=-10}^{10}T^{(l)}_i\subseteq \bigcup_{i=-10}^{10}T^{(l+1)}_i.$
\end{lemma}

\begin{proof}

By definition, $A^{(l)}_{i}-2A^{(l+1)}_{i}=0$ or $1$, and $B^{(l)}_{i}-2B^{(l+1)}_{i}=0$ or $1$. For $C^{(l)},$ we can see similar result as Lemma~\ref{relation-cl} still holds.

\begin{equation*}
\begin{split}
A^{(l+1)}_{i_0} +B^{(l+1)}_{k-i_0} -C^{(l+1)}_{k}
&\ge A^{(l)}_{i_0}/2 -1/2 +B^{(l)}_{k-i_0}/2 -1/2 -C^{(l)}_{k}/2-7/2\\
&\ge \frac{1}{2}\left(A^{(l)}_{i_0} +B^{(l)}_{k-i_0}-C^{(l)}_{k}\right)-9/2.\\
A^{(l+1)}_{i_0} +B^{(l+1)}_{k-i_0} -C^{(l+1)}_{k}
&\le A^{(l)}_{i_0}/2 +B^{(l)}_{k-i_0}/2 -C^{(l)}_{k}/2 + 8/2\\
&\le \frac{1}{2}\left(A^{(l)}_{i_0} +B^{(l)}_{k-i_0}-C^{(l)}_{k}\right)+4.\\
\end{split}
\end{equation*}

Therefore, when $-10\le A^{(l)}_{i_0} +B^{(l)}_{k-i_0} -C^{(l)}_{k}\le 10,$
\begin{equation*}
    -10<-10/2-9/2\le A^{(l+1)}_{i_0} +B^{(l+1)}_{k-i_0} -C^{(l+1)}_{k}\le 10/2+4<10.
\end{equation*}

\end{proof}

\begin{lemma}\label{conv-tb}
The expected number of segments in $T^{(l)}_b$ is $\tilde{O}(n^{2-\alpha}).$
\end{lemma}

\begin{proof}

When $2^l\geq p/100$, the total number of segments is bounded by $O(n^{2-\alpha})$, so next we assume that $2^l<p/100$.

For any segment $([i_0,i_1],k)$ of finite elements where $A^*_{i_0}+B^*_{k-i_0}\neq C^*_{k}$. By Assumption~\ref{conv-basic-assumption}, $(C_{k}\mod p)<2p/3$, so we can get $|A_{i_0}+B_{k-i_0}- C_{k}|\geq p/3$ as in Lemma~\ref{mono-tb}.

We want to bound the probability that $([i_0,i_1],k)$ appears in $T^{(l)}_{b}.$ If it is in $T^{(l)}_{b}$,
\begin{equation*}
\left\lfloor \frac{A_{i_0} \mod p}{2^l} \right\rfloor +\left\lfloor \frac{B_{k-i_0} \mod p}{2^l}\right\rfloor=C^{(l)}_{k}+b.
\end{equation*}

So $$-4\le \frac{A_{i_0} \mod p}{2^l}+\frac{B_{k-i_0} \mod p}{2^l}-\frac{C_{k} \mod p}{2^l}-b\le 4$$ 
%Note that $(A_{i_0} \mod p)+(B_{k-i_0} \mod p)-(C_{k} \mod p)-((A_{i_0} +B_{k-i_0}-C_{k})\mod p)= 0$ or $\pm p.$  
Let $C_{k}=A_{q}+B_{k-q}$, and
\begin{equation*}
    (A_{i_0} +B_{k-i_0}-A_{q}-B_{k-q})\mod p\in[2^l (b-4),2^l (b+4)].
\end{equation*}

That is, $A_{i_0} +B_{k-i_0}-A_{q}-B_{k-q}$ should be congruent to one of the $O\left(2^l\right)$ remainders.
As the argument in Lemma~\ref{mono-tb}, the probability that it falls into the range of length $O(2^l)$ is $O(2^l/n^{\alpha})$. Since the number of segments is $O(n^2/2^l)$, the expected number of segments in $T^{(l)}_b$ is $\tilde{O}(n^{2-\alpha})$.

\begin{lemma}
If $A_{i}+B_{k-i}=C_{k},$ then $A^{(l)}_{i}+B^{(l)}_{k-i}=C^{(l)}_{k}+b$ for some $-10\le b\le 10.$
\end{lemma}

\begin{proof}
By Assumption~\ref{conv-basic-assumption},
\begin{equation*}
\begin{split}
A^{(l)}_{i}+B^{(l)}_{k-i}-C^{(l)}_{k}&=\left\lfloor \frac{A_{i}\mod p}{2^l}\right\rfloor+\left\lfloor \frac{B_{k-i}\mod p}{2^l}\right\rfloor-C^{(l)}_{k}\\
&\le  \frac{A_{i}\mod p}{2^l}+ \frac{B_{k-i}\mod p}{2^l}- \frac{C_{k}\mod p}{2^l}+3\\
&=\frac{(A_{i}+B_{k-i}-C_{k})\mod p}{2^l}+3=3.\\
A^{(l)}_{i}+B^{(l)}_{k-i}-C^{(l)}_{k}&=\left\lfloor \frac{A_{i}\mod p}{2^l}\right\rfloor+\left\lfloor \frac{B_{k-i}\mod p}{2^l}\right\rfloor-C^{(l)}_{k}\\
&\ge  \frac{A_{i}\mod p}{2^l}+ \frac{B_{k-i}\mod p}{2^l}- \frac{C_{k}\mod p}{2^l}-4\\
&=\frac{(A_{i}+B_{k-i}-C_{k})\mod p}{2^l}-4=-4.\\
\end{split}
\end{equation*}

\end{proof}

Next we argue that our algorithm correctly computes all entries $C^{(l)}_{k}$ from $C^{(l+1)}_{k}$ and $T^{(l+1)}_b$, for $l=h-1,\cdots, 0$. Let $q$ be the index such that $C_{k} = A_{q} + B_{k-q}$. By the above lemma, there exists an integer offset $b\in [-10,10]$ such that $A^{(l+1)}_{q} + B^{(l+1)}_{k-q} = C^{(l+1)}_{k} + b$. Therefore, by construction of polynomials $A^p, B^p$, we have:
$$\begin{aligned}
	C^p_{k,b}(x) &= \sum_{i\mid A^{(l+1)}_{i} + B^{(l+1)}_{k-i} = C^{(l+1)}_{k} + b} x^{A^{(l)}_{i} - 2 A^{(l+1)}_{i} + B^{(l)}_{k-i} - 2 B^{(l+1)}_{k-i}} \\
    &= \sum_{i\mid (A^*_{i}+B^*_{k-i}=C^*_{k}) \land (A^{(l+1)}_{i} + B^{(l+1)}_{k-i} = C^{(l+1)}_{k} + b)} x^{A^{(l)}_{i} - 2 A^{(l+1)}_{i} + B^{(l)}_{k-i} - 2 B^{(l+1)}_{k-i}}\\
    &+    \sum_{i\mid (A^*_{i}+B^*_{k-i}\neq C^*_{k}) \land (A^{(l+1)}_{i} + B^{(l+1)}_{k-i} = C^{(l+1)}_{k} + b)} x^{A^{(l)}_{i} - 2 A^{(l+1)}_{i} + B^{(l)}_{k-i} - 2 B^{(l+1)}_{k-i}} \\
    &= \sum_{i\mid (A^*_{i}+B^*_{k-i}=C^*_{k}) \land (A^{(l+1)}_{i} + B^{(l+1)}_{k-i} = C^{(l+1)}_{k} + b)} x^{A^{(l)}_{i} - 2 A^{(l+1)}_{i} + B^{(l)}_{k-i} - 2 B^{(l+1)}_{k-i}} \\
    &+ \sum_{([i_0,i_1],k)\in T^{(l+1)}_b, i\in [i_0,i_1]} x^{A^{(l)}_{i} - 2 A^{(l+1)}_{i} + B^{(l)}_{k-i} - 2 B^{(l+1)}_{k-i}} \\
    &= x^{-2(C^{(l+1)}_{k}+b)}\cdot\sum_{i\mid (A^*_{i}+B^*_{k-i}=C^*_{k}) \land (A^{(l+1)}_{i} + B^{(l+1)}_{k-i} = C^{(l+1)}_{k} + b)} x^{A^{(l)}_{i} + B^{(l)}_{k-i}} + R^p_{k,b}(x) 
\end{aligned}$$
Since $A^*_{q}+B^*_{k-q}=C^*_{k}$ and $A^{(l+1)}_{q} + B^{(l+1)}_{k-q} = C^{(l+1)}_{k} + b$, when we extract $A^{(l)}_{q} + B^{(l)}_{k-q}$ from terms of $C_{k, b}^p(x) - R_{k, b}^p(x)$, it satisfies
$$\left\lfloor \frac{A_{q}\mod p}{2^l}\right\rfloor+\left\lfloor \frac{B_{k-q}\mod p}{2^l}\right\rfloor \leq \frac{(A_{q}+B_{k-q})\mod p}{2^l} =
\frac{C_{k}\mod p}{2^l}\leq \left\lfloor\frac{(C_{k}\mod p) +2^l-1}{2^l}\right\rfloor$$
$$\left\lfloor \frac{A_{q}\mod p}{2^l}\right\rfloor+\left\lfloor \frac{B_{k-q}\mod p}{2^l}\right\rfloor \geq \frac{((A_{q}+B_{k-q})\mod p)-2(2^l-1)}{2^l} \geq \left\lfloor\frac{(C_{k}\mod p)-2(2^l-1)}{2^l}\right\rfloor$$

Thus the term which gives $A^{(l)}_{q} + B^{(l)}_{k-q}$ can give a valid $C^{(l)}_{k}$. Also for every term which gives $A^{(l)}_{i} + B^{(l)}_{k-i}$ which satisfies $A^*_{i}+B^*_{k-i}=C^*_{k}$ and $A_{i} +B_{k-i}\geq C_{k}$, 
$$\left\lfloor \frac{A_{i}\mod p}{2^l}\right\rfloor+\left\lfloor \frac{B_{k-i}\mod p}{2^l}\right\rfloor \geq \frac{((A_{i}+B_{k-i})\mod p)-2(2^l-1)}{2^l} \geq \left\lfloor\frac{(C_{k}\mod p)-2(2^l-1)}{2^l}\right\rfloor$$
So by choosing the minimum, we can get a valid $C^{(l)}_{k}$.

\end{proof}

\section*{Acknowledgment}
Tianyi Zhang is supported by funding from the European Research Council (ERC) under the European Union’s Horizon 2020 research and innovation programme (grant agreement No 803118 UncertainENV).

\vspace{5mm}
\bibliographystyle{alpha}
\bibliography{ref}

\appendix
\section{When $B$ is column monotone}

In Section~\ref{monotone-main}, we consider the restricted case that the rows of $B$ are monotone. Now we explain how to calculate the min-plus product with the same asymptotic time complexity when $B$ is column-monotone, via minor adjustments of the recursive algorithm.

We want to calculate $C=A\star B,$ where $A,B$ are $n\times n$ matrices, and the columns of $B$ are monotonously non-decreasing. We can assume without loss of generality that the rows of $A$ are monotonously non-increasing: If there exists two entries $A_{i,k_1}$ and $A_{i,k_2}$ in the same row of $A$, with $k_1<k_2$ and $A_{i,k_1}<A_{i,k_2},$ then for any entry $C_{i,j}=\min_{k}\{A_{i,k}+B_{k,j}\},$ we have $A_{i,k_1}+B_{k_1,j}<A_{i,k_2}+B_{k_2,j},$ so the value of $A_{i,k_2}$ is never considered in the calculation, thus in this case we can set $A_{i,k_2}\leftarrow A_{i,k_1}$. When $B$ is bounded by $O(n),$ we can make $A$ and $C$ also bounded by $O(n)$ by the method in Section~\ref{sec:pre}

\iffalse
The change we should make on the recursive algorithm is the organization of segments: instead of fixing $i,k,$ we fix $i,j.$ We define a \emph{segment} w.r.t. $A^{(l)}$ and $B^{(l)}$ as $(i,j,[k_0,k_1]),$ where $i,j,k_0,k_1\in [n]$ satisfies that for all $k_0\leq k\leq k_1$, $B^{(l)}_{k,j}=B^{(l)}_{k_0,j}$, $B^*_{k,j}=B^*_{k_0,j}$ and $A^{(l)}_{i,k}=C^{(l)}_{i,k_0}$, $A^*_{i,k}=C^*_{i,k_0}$.

The other parts work just as that recursive algorithm. It is still correct because in the previous proof of correctness, we in fact prove it for each triple. The running time is still bounded by $\tilde{O}(n^{3-\alpha}+n^{\omega+\alpha})=\tilde{O}(n^{(3+\omega)/2})$: lemma \ref{app-mono-tb} holds after minor adjustments.
\fi

Let $\alpha\in (0, 1)$ be a constant parameter to be determined later, and pick a uniformly random prime number $p$ in the range of $[40n^\alpha, 80n^\alpha]$. Without loss of generality, let us assume that $n$ is a power of $2$. Next we make the following assumption about elements in $A$ and $B$:

\begin{assumption}\label{app-basic-assumption}
For every $i,j$, either $(A_{i,j}\mod p)<p/3$ or $A_{i,j}=+\infty$, and each row of $A$ is monotone besides the infinite elements. Similar for $B$: either $(B_{i,j}\mod p)<p/3$ or $B_{i,j}=+\infty$, and each column of $B$ is monotone besides the infinite elements.
\end{assumption}

By the same method in Lemma~\ref{conv-assumption}, we can prove:

\begin{lemma}\label{app-assumption-correctness}
The general computation of $A\diamond B$ can be reduced to a constant number of computations of $A^i\diamond B^i$ where all of $A^i,B^i$'s satisfy Assumption~\ref{app-basic-assumption}.
The number of intervals of infinity in each row of $A^i$ and in each column of $B^i$ is bounded by $O(n^{1-\alpha})$.
\end{lemma}

Define integer $h$ such that $2^{h-1}\leq p< 2^h$. For each integer $0\leq l\leq h$, let $A^{(l)}$ be the $n\times n$ matrix defined as $A^{(l)}_{i, j} = \lfloor \frac{A_{i,j}\mod p}{ 2^{l}} \rfloor$ if $A_{i,j}$ is finite, otherwise $A^{(l)}_{i,j}=+\infty$, similarly define matrix $B^{(l)}$. 

We will recursively calculate $C^{(l)}$ for $l=h,h-1,\cdots,0$, and if $C_{i,j}$ is finite, $C^{(l)}$ will satisfy 
$$\lfloor \frac{(C_{i,j} \mod p)-2(2^l-1)}{ 2^{l}} \rfloor \leq C^{(l)}_{i, j} \leq \lfloor \frac{(C_{i,j}\mod p)+2(2^l-1)}{ 2^{l}} \rfloor$$  (Note that $C^{(l)}$ is not necessarily equal to $A^{(l)}\star B^{(l)}$.) In the end when $l=0$ we can get the matrix $C^{(0)}_{i,j}=C_{i,j}\mod p,$ by the procedure of recursion. Define $A^*$ and $B^*$ as $A^*_{i,j}=\lfloor A_{i,j}/p\rfloor$ and $B^*_{i,j}=\lfloor B_{i,j}/p\rfloor$. We use the trivial method which checks each interval on $i$-th row of $A^*$ and $j$-th column of $B^*$ to calculate $C^*=A^*\star B^*$ in $\tilde{O}(n^{3-\alpha})$ time. By Assumption~\ref{app-basic-assumption}, $C^*_{i,j}=\lfloor C_{i,j}/p\rfloor$ if $C_{i,j}$ is finite. Thus we can calculate the exact value of $C_{i,j}$ by the result of $C_{i,j}\mod p.$

We can see all elements in $A^{(l)}, B^{(l)}, C^{(l)}$ are non-negative integers at most $O(n^\alpha/2^l)$ or infinite. Since $A$ is row-monotone and $B$ is column-monotone, every row of $A^{(l)}$ and every column of $B^{(l)}$ is composed of $O(n/2^l)$ intervals, where all elements in each interval are the same. The change we should make on the recursive algorithm is the organization of segments: instead of fixing $i,k,$ we fix $i,j.$
\begin{definition}
A \emph{segment} w.r.t. $A^{(l)}$ and $B^{(l)}$ as $(i,j,[k_0,k_1]),$ where $i,j,k_0,k_1\in [n]$ satisfies that for all $k_0\leq k\leq k_1$, $A_{i,k_0}$ and $B_{k_0,j}$ are finite, $A^{(l)}_{i,k}=A^{(l)}_{i,k_0}$ and $A^*_{i,k}=A^*_{i,k_0}$, $B^{(l)}_{k,j}=B^{(l)}_{k_0,j}$ and $B^*_{k,j}=B^*_{k_0,j}$.
\end{definition}
Then for the $i$-th row of $A^{(l)}$ and the $j$-th column of $B^{(l)}$, $[n]$ can be divided into $O(n/2^l)$ segments.

We maintain the auxiliary sets $T^{(l)}_b$ for $-10\le b\le 10$ throughout the algorithm, where the set $T^{(l)}_b$ consists of all the segments $(i,j,[k_0,k_1])$ w.r.t. $A^{(l)}$ and  $B^{(l)}$ satisfying: $$\text{$A_{i,k_0}$ is finite and $A^*_{i,k_0}+B^*_{k_0,j}\neq C^*_{i,j}$ and $A^{(l)}_{i,k_0}+B^{(l)}_{k_0,j}= C^{(l)}_{i,j}+b$ }$$

The algorithm proceeds as:
\begin{itemize}
    \item In the first iteration $l=h$,  $A^{(h)},B^{(h)},C^{(h)}$ are zero matrices, and it is easy to see $|T^{(h)}_b|=O(n^{3-\alpha})$.
    \item For $l=h-1,\cdots, 0$, we first compute $C^{(l)}$ with the help of $T^{(l+1)}_b$, then construct $T^{(l)}_b$ from $T^{(l+1)}_b$. By Lemma~\ref{app-mono-triple} that $\bigcup_{i=-10}^{10}T^{(l)}_i\subseteq \bigcup_{i=-10}^{10}T^{(l+1)}_i$, we can search the shorter segments contained in $T^{(l+1)}_b$ to find $T^{(l)}_b$. By Lemma~\ref{app-mono-tb}, $|T^{(l)}_b|$ is always bounded by $O(n^{3-\alpha})$. 
\end{itemize}

Each iteration has three phases:

\noindent\textbf{Polynomial matrix multiplication. }%\footnote{Change}
Construct two polynomial matrices $A^p$ and $B^p$ on variables $x, y$ in the following way:
When $A_{i,k}$ is finite,
$$A^p_{i, k} = x^{A^{(l)}_{i, k} - 2A^{(l+1)}_{i, k}}\cdot y^{ A^{(l+1)}_{i, k}}$$
Otherwise $A^p_{i,k}=0$, and when $B_{k,j}$ is finite,
$$B^p_{k, j} = x^{B^{(l)}_{k, j} - 2B^{(l+1)}_{k, j}}\cdot y^{B^{(l+1)}_{k, j}}$$ Otherwise $B^p_{k,j}=0$.
Then, compute the standard $(+, \times)$ matrix multiplication $C^p = A^p\cdot B^p$ using fast matrix multiplication algorithms. Note that $A^{(l)}_{i, j}-2A^{(l+1)}_{i, j}, B^{(l)}_{i, j}-2B^{(l+1)}_{i, j}$ are $0$ or $1,$ so the degree of $x$ terms are 0 or 1. This phase runs in time $\tilde{O}(n^{\omega+\alpha}).$ 

\noindent\textbf{Subtracting erroneous terms. }%\footnote{Change}
This phase is to extract the true values $C^{(l)}_{i, j}$'s from $C^{(l+1)}_{i, j}.$ The algorithm iterates over all offsets $-10\le b\le 10$, and enumerates all the segments in $T_b^{(l+1)}.$

For each pair of indices $i, j\in [n]$, if $C^p_{i,j}=0$ then $C^{(l)}_{i,j}=+\infty$, otherwise collect all the monomials $\lambda x^cy^d$ of $C_{i, j}^p$ such that 
$$d= {C}^{(l+1)}_{i, j} + b$$ 
and let $C_{i, j, b}^p(x)$ be the sum of all such terms $\lambda x^c$. Next, compute a polynomial 
$$R_{i, j, b}^p(x) = \sum_{(i, j, [k_0,k_1])\in T^{(l+1)}_b, k\in [k_0,k_1]} x^{A^{(l)}_{i, k} - 2A^{(l+1)}_{i, k} + B^{(l)}_{k,j} - 2B^{(l+1)}_{k,j}}$$
Finally, let $s_{i, j, b}$ be the minimum degree of $x$ in the polynomial $C^p_{i, j, b}(x) - R^p_{i, j, b}(x)$, and compute a candidate value $c_{i, j, b} = 2d + s_{i, j, b} $. Ranging over all integer offsets $-10\le b\le 10$, take the minimum of all candidate values and output as $C^{(l)}_{i, j} = \min_{-10\le b\le 10}\{c_{i, j, b}\}$. This phase runs in time $\tilde{O}(n^{2+\alpha}+n^{3-\alpha})$ (see Lemma~\ref{app-mono-tb}), since every segment $(i, j, [k_0,k_1])\in T^{(l+1)}_b$ contains at most two different $A^{(l)}_{i,k}$ and two different $B^{(l)}_{k,j}$, thus it is easy to compute all of $C^p_{i, j, b}(x) - R^p_{i, j, b}(x)$ in $O(n^{2+\alpha}+|T^{(l+1)}_b|)$ time.

\noindent\textbf{Computing Triples $T^{(l)}_b$.} 
Since $A^{(l)}_{k,j}-2A^{(l+1)}_{k,j}$ and $B^{(l)}_{i,j}-2B^{(l+1)}_{i,j}$ are both 0 or 1, so each segment w.r.t. $A^{(l+1)}, B^{(l+1)}$ can be split into at most $O(1)$ segments w.r.t. $A^{(l)}, B^{(l)}$. 
%To compute $T^{(l)}_b$, initially set all $T^{(l)}_b\leftarrow \emptyset$ for all $|b|\leq 10$. 
By Lemma \ref{app-mono-triple} we know that $\bigcup_{i=-10}^{10}T^{(l)}_i$ is contained in $\bigcup_{i=-10}^{10}T^{(l+1)}_i,$ so our work here is to check the sub-segments of each segment in $\bigcup_{i=-10}^{10}T^{(l+1)}_i$ and put it into the $T^{(l)}_b$ it belongs to. This phase runs in time $\tilde{O}(|T^{(l+1)}_b|).$

The expected running time of the recursive algorithm is bounded by $\tilde{O}(n^{3-\alpha} + n^{\omega+\alpha})$ by Lemma~\ref{app-mono-tb}. Taking $\alpha = (3 - \omega) / 2$, the running time becomes $\tilde{O}(n^{(3+\omega) /2})$.

\subsection{Proof of correctness}

We can get a similar lemma as Lemma~\ref{relation-cl},
\begin{lemma}\label{app-relation-cl}
In each iteration $l=h-1,\cdots,0$, $-7\leq C^{(l)}_{i,j}-2C^{(l+1)}_{i,j}\leq 8$.
\end{lemma}

\begin{lemma}\label{app-mono-triple}
%$T^{(l)}_b \subseteq \bigcup_{i=\lfloor b/2-1\rfloor}^{\lfloor (b+1)/2\rfloor}T^{(l+1)}_i,$ for each $-10\le b\le 10.$ Therefore 
We have $\bigcup_{i=-10}^{10}T^{(l)}_i\subseteq \bigcup_{i=-10}^{10}T^{(l+1)}_i$, that is, the segments we consider in each iteration must be sub-segments of the segments in the last iteration.
    
\end{lemma}

\begin{proof}

Segments $(i,j,[k_0,k_1])$ in $T^{(l)}_b$ and $T^{(l+1)}_b$ must satisfy $A_{i,k_0},B_{k_0,j}$ are finite and $A^*_{i,k_0}+B^*_{k_0,j}\neq C^*_{i,j}$. By definition, $A^{(l)}_{i,k_0}-2A^{(l+1)}_{i,k_0}=0$ or $1$, and similar for $B$. By Lemma~\ref{app-relation-cl}, we have

\begin{equation*}
\begin{split}
A^{(l+1)}_{i, k} +B^{(l+1)}_{k, j} -C^{(l+1)}_{i, j}
&\ge \frac{1}{2}\left(A^{(l)}_{i, k} +B^{(l)}_{k, j}-C^{(l)}_{i, j}\right)-9/2.\\
A^{(l+1)}_{i, k} +B^{(l+1)}_{k, j} -C^{(l+1)}_{i, j}
&\le \frac{1}{2}\left(A^{(l)}_{i, k} +B^{(l)}_{k, j}-C^{(l)}_{i, j}\right)+4.\\
\end{split}
\end{equation*}

Therefore, when $-10\le A^{(l)}_{i, k} +B^{(l)}_{k, j} -C^{(l)}_{i, j}\le 10,$
\begin{equation*}
    -10<-10/2-9/2\le A^{(l+1)}_{i, k} +B^{(l+1)}_{k, j} -C^{(l+1)}_{i, j}\le 10/2+4<10.
\end{equation*}

\end{proof}

\begin{lemma}\label{app-mono-tb}
The expected number of segments in $T^{(l)}_b$ is $\tilde{O}(n^{3-\alpha}).$
\end{lemma}

\begin{proof}
As before we assume that $2^l<p/100$. For any segment $(i,j,[k_0,k_1])$ and $k \in [k_0, k_1]$ where $A_{i,k}, B_{k,j}$ are finite and $A^*_{i,k}+B^*_{k,j}\neq C^*_{i,j}$, similar to proof in Lemma~\ref{mono-tb}, we get $|A_{i,k}+B_{k,j}- C_{i,j}|\geq p/3$.

We want to bound the probability that $(i,j,[k_0,k_1])$ appears in $T^{(l)}_{b}.$ If it is in $T^{(l)}_{b}$,
\begin{equation*}
\left\lfloor \frac{A_{i,k} \mod p}{2^l} \right\rfloor +\left\lfloor \frac{B_{k,j} \mod p}{2^l}\right\rfloor=C^{(l)}_{i,j}+b.
\end{equation*}

So $$-4\le \frac{A_{i,k} \mod p}{2^l}+\frac{B_{k,j} \mod p}{2^l}-\frac{C_{i,j} \mod p}{2^l}-b\le 4$$ 
%Note that $(A_{i,k} \mod p)+(B_{k,j} \mod p)-(C_{i,j} \mod p)-((A_{i,k} +B_{k,j}-C_{i,j})\mod p)= 0$ or $\pm p.$
\begin{equation*}
    (A_{i,k} +B_{k,j}-C_{i,j})\mod p\in[2^l (b-4),2^l (b+4)].
\end{equation*}

That is, $A_{i,k} +B_{k,j}-C_{i,j}$ should be congruent to one of the $O\left(2^l\right)$ remainders. For each possible remainder $r\in[2^l(b-4),2^l(b+4)],$ ($|b|\leq 10$), we have
\begin{equation*}
    |r|\le 14\cdot 2^l<p/6\leq\frac{1}{2}\mid A_{i,k} +B_{k,j}-C_{i,j}\mid.
\end{equation*}

So $|(A_{i,k} +B_{k,j}-C_{i,j})-r|$ is a positive number bounded by $O(n)$, and the number of different primes $p\in[40n^\alpha,80n^\alpha]$ that $p\mid (A_{i,k} +B_{k,j}-C_{i.k})-r$ can not exceed $1/\alpha=O(1).$ In our algorithm, when we uniformly choose a prime $p$ from $[40n^\alpha,80n^\alpha],$ the probability that $(A_{i,k} +B_{k,j}-C_{i,j})\mod p =r$ is $\tilde{O}\left(\frac{1}{n^\alpha}\right).$ Since there are $O(2^l)$ such possible remainders, in expectation we have $O(2^l)\cdot O\left(\frac{n^{3}}{2^l}\right)\cdot\tilde{O}\left(\frac{1}{n^\alpha}\right)=\tilde{O}(n^{3-\alpha})$ segments in $T^{(l)}_b.$

\end{proof}

From the proof of Lemma~\ref{exact-approx}, we can get:

\begin{lemma}
If $A_{i,k}+B_{k,j}=C_{i,j},$ then $A^{(l)}_{i,k}+B^{(l)}_{k,j}=C^{(l)}_{i,j}+b$ for some $-10\le b\le 10.$
\end{lemma}

Next we argue that our algorithm correctly computes all entries $C^{(l)}_{i, j}$ from $C^{(l+1)}_{i, j}$ and $T^{(l+1)}_b$, for $l=h-1,\cdots, 0$. Let $q$ be the index such that $C_{i, j} = A_{i, q} + B_{q, j}$. By the above lemma, there exists an integer offset $b\in [-10,10]$ such that $A^{(l+1)}_{i, q} + B^{(l+1)}_{q, j} = C^{(l+1)}_{i, j} + b$. Therefore, by construction of polynomial matrices $A^p, B^p$, we have:
$$\begin{aligned}
	C^p_{i, j, b}(x) 
    &= \sum_{k\mid (A^*_{i,k}+B^*_{k,j}=C^*_{i,j}) \land (A^{(l+1)}_{i, k} + B^{(l+1)}_{k, j} = C^{(l+1)}_{i, j} + b)} x^{A^{(l)}_{i, k} - 2 A^{(l+1)}_{i, k} + B^{(l)}_{k, j} - 2 B^{(l+1)}_{k, j}}\\
    &+    \sum_{k\mid (A^*_{i,k}+B^*_{k,j}\neq C^*_{i,j}) \land (A^{(l+1)}_{i, k} + B^{(l+1)}_{k, j} = C^{(l+1)}_{i, j} + b)} x^{A^{(l)}_{i, k} - 2 A^{(l+1)}_{i, k} + B^{(l)}_{k, j} - 2 B^{(l+1)}_{k, j}} \\
    &= \sum_{k\mid (A^*_{i,k}+B^*_{k,j}=C^*_{i,j}) \land (A^{(l+1)}_{i, k} + B^{(l+1)}_{k, j} = C^{(l+1)}_{i, j} + b)} x^{A^{(l)}_{i, k} - 2 A^{(l+1)}_{i, k} + B^{(l)}_{k, j} - 2 B^{(l+1)}_{k, j}} \\
    &+ \sum_{(i,j,[k_0,k_1])\in T^{(l+1)}_b, k\in [k_0,k_1]} x^{A^{(l)}_{i, k} - 2 A^{(l+1)}_{i, k} + B^{(l)}_{k, j} - 2 B^{(l+1)}_{k, j}} \\
    &= x^{-2(C^{(l+1)}_{i,j}+b)}\cdot\sum_{k\mid (A^*_{i,k}+B^*_{k,j}=C^*_{i,j}) \land (A^{(l+1)}_{i, k} + B^{(l+1)}_{k, j} = C^{(l+1)}_{i, j} + b)} x^{A^{(l)}_{i, k} + B^{(l)}_{k, j}} + R^p_{i,j,b}(x) 
\end{aligned}$$
Since $A^*_{i,q}+B^*_{q,j}=C^*_{i,j}$ and $A^{(l+1)}_{i, q} + B^{(l+1)}_{q, j} = C^{(l+1)}_{i, j} + b$, when we extract $A^{(l)}_{i, q} + B^{(l)}_{q, j}$ from terms of $C_{i, j, b}^p(x) - R_{i, j, b}^p(x)$, it satisfies
$$\left\lfloor \frac{A_{i,q}\mod p}{2^l}\right\rfloor+\left\lfloor \frac{B_{q,j}\mod p}{2^l}\right\rfloor \leq \frac{(A_{i,q}+B_{q,j})\mod p}{2^l} =
\frac{C_{i,j}\mod p}{2^l}\leq \left\lfloor\frac{(C_{i,j}\mod p) +2^l-1}{2^l}\right\rfloor$$
$$\left\lfloor \frac{A_{i,q}\mod p}{2^l}\right\rfloor+\left\lfloor \frac{B_{q,j}\mod p}{2^l}\right\rfloor \geq \frac{((A_{i,q}+B_{q,j})\mod p)-2(2^l-1)}{2^l} \geq \left\lfloor\frac{(C_{i,j}\mod p)-2(2^l-1)}{2^l}\right\rfloor$$

Thus the term which gives $A^{(l)}_{i, q} + B^{(l)}_{q, j}$ can give a valid $C^{(l)}_{i,j}$. Also for every term which gives $A^{(l)}_{i, k} + B^{(l)}_{k, j}$ which satisfies $A^*_{i,k}+B^*_{k,j}=C^*_{i,j}$ and $A_{i, k} +B_{k,j}\geq C_{i.j}$, 
$$\left\lfloor \frac{A_{i,k}\mod p}{2^l}\right\rfloor+\left\lfloor \frac{B_{k,j}\mod p}{2^l}\right\rfloor \geq \frac{((A_{i,k}+B_{k,j})\mod p)-2(2^l-1)}{2^l} \geq \left\lfloor\frac{(C_{i,j}\mod p)-2(2^l-1)}{2^l}\right\rfloor$$
So by choosing the minimum, we can get a valid $C^{(l)}_{i,j}$.

\end{document}